\RequirePackage{amsmath,amssymb}
\documentclass[runningheads]{llncs}
\usepackage{xcolor}
\usepackage{graphicx}
\usepackage{paralist}
\usepackage{orcidlink}
\usepackage{hyperref}
\usepackage{cleveref}

\usepackage{lmodern}

\usepackage{tikz}
\usetikzlibrary{arrows}
\usetikzlibrary{positioning}
\usepackage{listings}
\usepackage[T1]{fontenc}
\usepackage[scaled=0.85]{beramono}

\usepackage[firstpage]{draftwatermark} %
\usepackage{etoolbox} %

\newbool{artifactfunctional}
\newbool{artifactreusable}
\newbool{artifactavailable}

\setbool{artifactfunctional}{true}
\setbool{artifactreusable}{true}
\setbool{artifactavailable}{true}
\newcommand{\artifacturl}{https://doi.org/10.5281/zenodo.12664114}
\newcommand{\artifactvposition}{14.2cm}

\SetWatermarkAngle{0}
\SetWatermarkText{\raisebox{\artifactvposition}{%
\hspace{0.1cm}%
\ifbool{artifactavailable}{%
\href{\artifacturl}{\includegraphics[width=2cm]{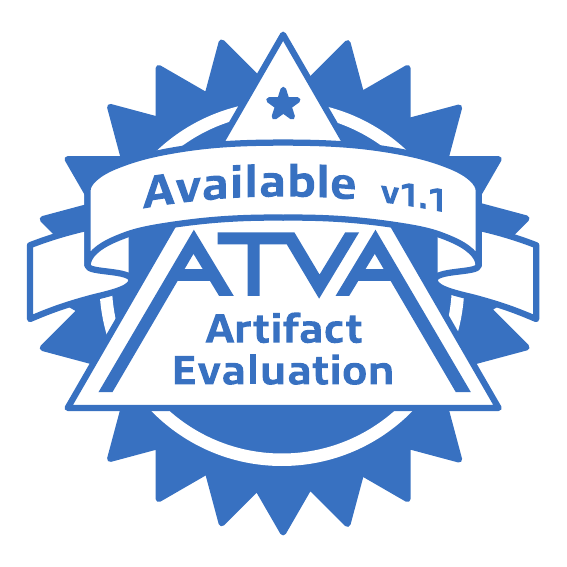}}}%
{\hspace{2cm}}%
\hspace{8.3cm}%
\ifbool{artifactfunctional}{%
\ifbool{artifactreusable}%
{\href{\artifacturl}{\includegraphics[width=2cm]{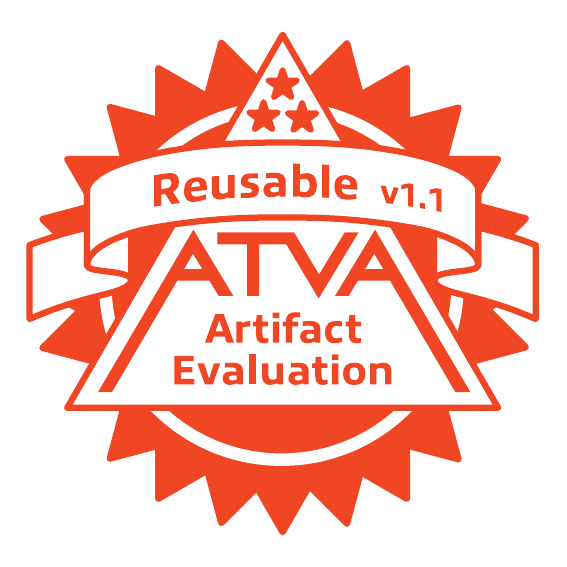}}}%
{\href{\artifacturl}{\includegraphics[width=2cm]{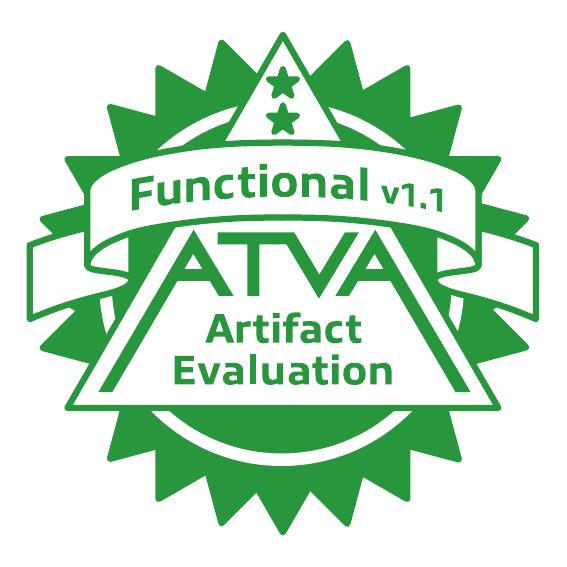}}}}%
{\hspace{2cm}}%
}}

\begin{document}

\newif\ifcomments

\commentstrue

\newif\ifconference
\conferencefalse

\ifcomments
\newcommand{\sharon}[1]{\textcolor{red}{SH: #1 }}
\newcommand{\raz}[1]{\textcolor{blue}{RL: #1 }}
\newcommand{\eden}[1]{{\color{purple}{EF: #1 }}}
\else
\newcommand{\sharon}[1]{}
\newcommand{\raz}[1]{}
\newcommand{\eden}[1]{}
\fi

\newcommand{\commentout}[1]{}
\newcommand{\mypara}[1]{\medskip \noindent \emph{#1}}

\newcommand{\updatename}[0]{high-low update}

\newcommand{\simul}[0]{H}
\newcommand{\idle}[0]{\mathrm{id}_\Sigma}
\newcommand{\sq}[0]{\eta}
\newcommand{\cond}[0]{\theta}
\newcommand{\formula}[0]{\alpha}
\newcommand{\formulaa}[0]{\beta}
\newcommand{\hint}{\gamma}
\newcommand{\assign}{v}

\newcommand{\struct}{s}
\newcommand{\structset}{\mathrm{struct}}
\newcommand{\dom}{\mathrm{dom}}
\newcommand{\FV}{\textit{FV}}
\newcommand{\Vars}{\textit{Vars}}
\newcommand{\fol}{\mathrm{fol}}
\newcommand{\para}{\mathrm{fin}}
\newcommand{\Tspec}{\mathcal T}
\newcommand{\Tpar}{T_\para}
\newcommand{\Tfol}{T_\fol}
\newcommand{\init}[0]{I}
\newcommand{\inv}[0]{J}

\newcommand{\size}{\mathrm{size}_{> k}}

\newcommand{\nat}{\mathbb{N}}
\newcommand{\domain}[0]{\mathcal D}
\newcommand{\interp}[0]{\mathcal I}
\newcommand{\terminterp}[3]{{#1}^{#2,#3}}
\newcommand{\safetyspec}[0]{\varphi}
\newcommand{\safetysem}[0]{P}
\newcommand{\ADrep}[2]{[{#2}]_{\setminus #1}}
\newcommand{\exc}[1]{#1_{\times}}
\newcommand{\ellpr}{{\ell\prime}}
\newcommand{\hpr}{{h\prime}}
\newcommand{\seq}[1]{{\vec #1}}

\newcommand{\closed}[2]{\mathrm{closed}_{#2}({#1})}
\newcommand{\closure}[1]{\mathrm{closure}_\Sigma({#1})}
\newcommand{\sort}[1]{\mathsf{#1}}
\newcommand{\mypyvy}{\texttt{mypyvy}}

\title{Proving Cutoff Bounds 
for Safety Properties 
in First-Order Logic}

\author{Raz Lotan\,\orcidlink{0009-0008-5883-5082} \and
Eden Frenkel\,\orcidlink{0009-0009-4589-2173} \and
Sharon Shoham\,\orcidlink{0000-0002-7226-3526}}

\authorrunning{Raz Lotan, Eden Frenkel and Sharon Shoham}

\institute{
Tel Aviv University, Tel Aviv, Israel\\
\email{lotanraz@tauex.tau.ac.il}\\}

\maketitle            

\begin{abstract}
First-order logic has been established as an important tool for modeling and verifying intricate systems such as distributed protocols and concurrent systems.
These systems are parametric in the number of nodes in the network or the number of threads, which is finite in any system instance, but unbounded.
One disadvantage of first-order logic is that it cannot distinguish between finite and infinite structures, leading to spurious counterexamples. 
To mitigate this, we offer a %
verification approach that captures only finite system instances.
Our %
approach is an adaptation of the cutoff method to systems modeled in first-order logic.
The idea is 
to show that any safety violation in a system instance of size larger than some bound can be simulated by a safety violation in a system of a smaller size.
The simulation provides an inductive argument for correctness in finite instances,
reducing the problem to showing safety of  instances with bounded size.
To this end, we develop a framework to (i)~encode such  simulation relations in first-order logic
and to (ii)~validate the simulation relation by a set of verification conditions given to an SMT solver. 
We apply our approach to verify safety of a set of examples, some of which cannot be proven by a first-order inductive invariant.

\end{abstract}

\section{Introduction}\label{sec:introduction}

Many recent works~\cite{local_reasoning,VeriCon,bhat_nagar,phase_structures,updr,pdr_alterations,ivy,induction_duality,duoAI,distAI} have used first-order logic for deductive or automatic verification of distributed protocols, concurrent programs, and other systems that operate over unbounded domains. 
Most of these works are based on the notion of an \emph{inductive invariant} as a way of verifying 
safety properties --- properties
of all reachable states of a system.
An inductive invariant is a property that holds in the initial states of the system and is preserved by the steps of the system. 
The system, its properties and its inductive invariants are all specified using first-order formulas, which allows to 
build on the success of first-order provers and SMT solvers, with or without theories, %
for automatically discharging the verification conditions that establish correctness of the system.

Using first-order logic for modeling and verifying systems and their inductive invariants has significant benefits. First, first-order logic provides a natural and expressive specification language 
that, thanks to uninterpreted relation and function symbols as well as quantifiers, 
is able to  model %
unbounded domains, such as network topologies, sets of messages, quorums of processes, %
and so on. 
In fact, first-order logic is Turing-complete as a specification language for systems \cite{ivy}. 
Second, first-order logic benefits from a plethora of automatic solvers that already exhibit impressive performance and continue to improve rapidly.

Unfortunately, %
using first-order logic for encoding and verifying systems and their inductive invariants also has several crucial downsides.
One obstacle is that it is not always possible to express a suitable inductive invariant in first-order logic. 
Another problem arises from the inability of first-order logic
to distinguish between finite and infinite structures.
This may lead to spurious infinite counterexamples to induction~\cite{infinite_models}, or even to spurious violations of the property, in cases where 
all finite instances of the system satisfy a desired property, but some infinite instances do not
(see \Cref{ex1}).  
Because we are usually only interested in finite, albeit unbounded, instances (e.g., a finite number of processes), violations caused by infinite instances are spurious.

In this work, we 
mitigate the aforementioned obstacles by (i)~considering a \emph{finite-domain} semantics for first-order transition systems that only includes the finite instances of the system (of any size), and 
(ii)~devising a deductive approach 
for establishing a \emph{cutoff} for a first-order transition system and a safety property 
under the finite-domain semantics
via a \emph{size-reducing simulation} between instances of the system.
The simulation relation is encoded and validated using first-order logic.
This allows us to take the finite-domain semantics into account in the verification process, while still using first-order solvers to discharge the resulting verification conditions.

Cutoffs are widely used in the context of parameterized systems. 
A parameterized system is a family of transition systems $T=\{T_n\}_n$, one for each natural number $n\in \nat$. The parameter $n$ can represent the number of nodes in a distributed network, the number of threads running a concurrent program, the length of an array or a linked list, etc. The parameterized system $T$ satisfies a property $P$ if $T_n$ satisfies $P$ for all $n$. 
A number $k$ is a cutoff for the system $T$ and the property $P$ if, to verify that $T$ satisfies $P$, it suffices to verify that $T_n$ satisfies $P$ for all $n\leq k$. Cutoffs therefore allow to reduce the verification of the unbounded system $T$ to verification of a system where $n$ is bounded by $k$. 
The cutoff method, introduced by~\cite{reasoning_about_rings},
proposes to establish that $k$ is a cutoff for $T$ 
by proving that for every $n>k$,  $T_n$ is simulated by $T_{n'}$ for some $n' < n$.

Much like a parameterized system, a first-order specification of a transition system induces a family of transition systems, except that instead of a numeric parameter, the instances are defined by the domain of their states. 
We can therefore adapt the cutoff method to verify all the finite instances of a 
first-order transition system (i.e., instances obtained by fixing a finite domain):
instead of a 
simulation between any $T_n$ for $n>k$ and some $T_{n'}$ with $n' < n$, we need to show a simulation between any instance with a finite domain of size $>k$ and some instance with a strictly smaller domain.

The key technical contribution of our work is a realization of the cutoff method for first-order transition systems using a size-reducing simulation relation between finite instances of the system. 
A size-reducing simulation relation can be understood as proving the induction step 
in a proof by induction on the size of the domain.
It therefore %
lifts correctness of  instances of size at most $k$ to all finite instances, but \emph{not} to infinite instances.
As a result, it may succeed to verify first-order specifications that only satisfy their property w.r.t.\ the finite-domain semantics and not w.r.t.\ the first-order semantics.
Moreover, as already demonstrated in previous work~\cite{squeezers}, 
verification by simulation (akin to induction on domain size) may sometimes be easier than verification based on inductive invariants (akin to induction on time), and can succeed in cases where an inductive invariant cannot be expressed in first-order logic. 

We propose a first-order encoding of a size-reducing simulation relation 
for a first-order transition system. We further show how to verify the correctness of such a (user-provided) encoding by a set of verification conditions that are encoded as first-order formulas and discharged by an automatic solver.

Our first challenge is  to ensure that the simulation relation is size-reducing in finite domains. While first-order logic can express the property that the size of a set is exactly $k$, at most $k$ or at least $k$, for any specific number $k$, it cannot directly compare the sizes of sets of arbitrary sizes.
Our solution is to obtain size reduction 
by defining the smaller domain as a subset of the larger domain obtained by deleting one element.
The element to be deleted is specified by a (user-provided) formula.
For finite instances, this ensures that the size of the domain decreases.  

The next challenge arises from the definition of a simulation relation, which  involves quantification over states:
\emph{every} successor state of the larger state must be simulated by \emph{some} successor state of the smaller state. 
When states are modeled using uninterpreted functions and relations, this 
leads to higher-order quantification.
To facilitate a first-order encoding of the verification conditions, we strengthen the definition of the simulation relation.
Namely, we replace the existential quantification over a simulating successor state of the smaller state
by universal quantification, stating that \emph{every} state that simulates the larger successor state is a successor of the smaller state. To ensure that the universal quantification is not satisfied vacuously, we require the (inductive) totality of the simulation relation.%

\paragraph{Contributions.} The main contributions of this paper are:
\begin{itemize}
    \item We define the finite-domain semantics of first-order specifications of transition systems, formulate the notion of cutoffs, and define size-reducing simulations as a way of proving cutoffs for finite-domain first-order transition systems %
    (\Cref{sec:cutoff}).
    \item We develop an approach for defining and validating size-reducing simulation relations in first-order logic (\Cref{sec:smt_encoding}).
    \item We implement a tool that receives as input a first-order specification of a transition system and a safety property, as well as a first-order encoding of a size-reducing simulation relation, and validates the simulation relation using the Z3~\cite{z3} SMT solver (\Cref{sec:evaluation}). 
    \item We demonstrate the applicability of the approach by 
    verifying a set of examples of distributed protocols and concurrent systems, including some that cannot be verified by classical techniques based on inductive invariants (\Cref{sec:evaluation}).
\end{itemize}
The rest of the paper is organized as follows: \Cref{sec:prelims} provides the necessary background on transition systems in first-order logic, 
\Cref{sec:cutoff,sec:smt_encoding,sec:evaluation} present the above contributions, \Cref{sec:related} discusses related work and \Cref{sec:conclusion} concludes the paper. 
\ifconference
Due to space considerations we defer all proofs to the full version of this paper \cite{full_version}.
\else 
We defer all proofs to \Cref{appendix:proofs}. 
\fi

\section{Preliminaries} \label{sec:prelims}

This section provides a brief background on transition systems, first-order logic and first-order specifications of transition systems and properties.

\paragraph{Transition Systems.}
a transition system $T$ is given by $T=(S,\init,R)$, where $S$ is the set of states, $\init \subseteq S$ is the set of initial states, and $R\subseteq S\times S$ is the transition relation. 
A trace of $T$ is a sequence of states $(s_i)_{i=0}^t$ such that $s_0\in \init$ and for every $0\leq i<t$ we have $(s_i,s_{i+1})\in R$.
We denote by $R^\star$ the reflexive transitive closure of $R$.
A state $s$ is reachable if there exists a trace $(s_i)_{i=0}^t$ such that $s_t = s$. A safety property is given by a set of states $P \subseteq S$.
We say that $T$ satisfies $P$, denoted $T \models P$, if every reachable state of $T$ is in $P$.
A set of states $\inv \subseteq S$ is an inductive invariant of $T$ if for any state $s\in \init$ we have $s\in \inv$ and for any two states $(s,s')\in R$ such that $s\in \inv$ we have $s'\in \inv$.
Note that if $\inv$ is an inductive invariant and $\inv \subseteq P$ then $T\models P$.

\paragraph{First-Order Logic.} 
We consider first-order logic with equality. 
For simplicity of the presentation, we consider the single-sorted, uninterpreted version of first-order logic, i.e., without background theories. However, in practice, we use  many-sorted first-order logic and allow some sorts and symbols to be interpreted by arbitrary background theories.
See \Cref{subsec:many_sorted} for more details.

A first-order vocabulary $\Sigma$ consists of relation and function symbols, where 
function symbols with arity zero are also called constant symbols.
Terms $t$ over $\Sigma$ are either variables $x$ or function applications $f(t_1,\ldots,t_m)$.
Formulas are defined recursively, where atomic formulas are either $t_1=t_2$ or $r(t_1,\ldots,t_m)$ where $r$ is a relation symbol.
Non-atomic formulas are built using connectives $ \neg, \wedge, \vee, \to $ and quantifiers $\forall,\exists$.
We denote by $\Vars(\formula)$ the variables of a formula $\alpha$, by $\FV(\formula)$ the free variables of $\formula$, and write $\formula(x_1,\ldots,x_m)$ when $\FV(\formula)\subseteq \{x_1,\ldots,x_m\}$. 
We denote a sequence $x_1,\ldots,x_m$ by $\seq x$ when $m$ is clear from context.

A structure for $\Sigma$ is a pair $\struct = (\domain,\interp)$ 
where $\domain$ is a domain (a nonempty set of elements), denoted $\dom(\struct)$, 
and $\interp$ is an interpretation that maps each relation and function symbol to an appropriate construct over the domain.
We denote by $\structset(\Sigma)$ the set\footnote{\label{footnote-sets}%
Technically, $\structset(\Sigma)$ is not a well-defined set in the sense of set theory (since we allow the domain of a structure to be an arbitrary set).  
However, for the sake of this paper we can restrict the domains to be subsets of $\nat$
without loss of generality. }
of all structures for $\Sigma$.
We define the cardinality (size) of $\struct$, denoted $|\struct|$, to be $|\dom(\struct)|$. 
For two structures $s=(\domain,\interp)$ and $\exc s=(\exc \domain, \exc \interp)$, we say that $\exc s$ is a substructure of $s$ if $\exc \domain \subseteq \domain$, for every relation $r\in \Sigma$ with arity $m$ we have $\exc \interp(r) = \interp(r) \cap \exc \domain^m $ and for every function $f\in \Sigma$ with arity $m$ we have $\exc \interp(f) = \interp(f) |_{\exc \domain^m}$.
For a formula $\formula$ with free variables $\FV(\formula)$, a structure $\struct=(\domain,\interp) $ and an assignment $\assign:\FV(\formula) \to \domain$ mapping each free variable in $\formula$ to an element from the domain, we write $\struct,\assign\models \formula$ to denote that the pair $(s,\assign)$ satisfies $\formula$, defined as usual.
For a variable $x$ and a domain element $d$, we write $[x\to d]$ for the assignment $\assign:\{x\}\to \domain$ with $\assign(x)=d$.
For two formulas $\formula,\formulaa$ (not necessarily closed), we write $ \formula \implies \formulaa$ to denote the validity of $\formula \to \formulaa$.

For a vocabulary $\Sigma$ and an identifier $o$, we denote by $\Sigma^o$ the vocabulary $\{a^o \mid a \in \Sigma\}$. We assume that $\Sigma \cap \Sigma^o = \emptyset$. 
For a formula (or term) $\formula$ over $\Sigma$ we denote by $\formula^o$ the formula (respectively, term) over $\Sigma^o$ obtained by substituting each symbol $a \in \Sigma$ by $a^o \in \Sigma^o$.
With abuse of notation, we sometimes consider an interpretation $\interp$ for $\Sigma$ as an interpretation for $\Sigma^o$ where $\interp(a^o) = \interp(a)$, or, similarly, an interpretation $\interp^o$ of $\Sigma^o$ as an interpretation for $\Sigma$.
\paragraph{First-Order Specifications of  Transition Systems.}
We consider transition systems given by a first-order logic specification $\Tspec=(\Sigma,\Gamma,\iota,\tau)$ where  $\Sigma$ is a vocabulary, $\Gamma$ is a closed formula over $\Sigma$ that acts as a finitely-axiomatizable background theory, $\iota$ is a closed formula over $\Sigma$ that specifies initial states, and $\tau$ is a closed formula defined over vocabulary $\Sigma\uplus\Sigma'$ that specifies the transition relation. Intuitively, the unprimed copy of $\Sigma$ specifies the pre-state of a transition, while the primed copy specifies the post-state.
Some symbols in $\Sigma$ may be declared immutable, in which case $\tau$ preserves their interpretation.

Any domain $\domain$ defines a transition system $T_\domain=(S_\domain,\init_\domain,R_\domain)$ where 
$S_\domain$ is the set of structures $s\in \structset(\Sigma)$ with $\dom(s) = \domain$ that satisfy the axioms, i.e., $s\models \Gamma$;
$\init_\domain$ is the set of states $s\in S_\domain$ that satisfy $\iota$; and the transition relation $R_\domain$ is the set of all pairs of states $\left((\domain,\interp),(\domain,\interp')\right)$ where, when we identify $\interp$ as an interpretation of $\Sigma$ and $\interp'$ as an interpretation of $\Sigma'$ we have $(\domain,\interp\uplus \interp')\models \tau$.

The standard first-order semantics of $\Tspec$ is then given by the transition system $\Tfol=(S_\fol,\init_\fol,R_\fol)$ where $S_\fol,\init_\fol,R_\fol$ are the (disjoint) union of their counterparts in $T_\domain$ over all domains $\domain$ 
(see \cref{footnote-sets} for well-definedness of the union). 
Note that every trace of $\Tfol$ is a trace of $T_\domain$ for some domain $\domain$. This is because 
$(s,s')\in R_\fol$ implies $\dom(s)=\dom(s')$. 
We refer to $T_\domain$ as a system instance of $\Tfol$.

Safety properties and inductive invariants are also specified by closed first-order formulas over $\Sigma$. A safety specification $\safetyspec$ induces a safety property $\safetysem = \{ s \in S_\fol \mid s \models \safetyspec\}$ for $\Tfol$. With abuse of notation we will also consider $\safetysem$ as a safety property of $T_\domain$ for any domain $\domain$ instead of writing $\safetysem \cap S_\domain$.

\section{The Cutoff Method for First-Order Transition Systems}\label{sec:cutoff}

The standard first-order semantics $\Tfol$ of a transition system specification $\Tspec$  is often an abstraction of the intended semantics in the sense that it considers as a state \emph{any} first-order structure (that satisfies $\Gamma$), whereas in many cases the intended semantics only considers finite structures. 
In this section we define a finite-domain semantics $\Tpar$ for $\Tspec$, where the size of states (structures) is finite but unbounded. The finite-domain semantics induces a parameterized system, where the parameter is the size of the (finite) domain.
We then formulate the cutoff method, based on the notion of a size-reducing simulation, in the context of first-order transition systems, as a way of reducing safety verification of $\Tpar$ 
to safety verification of a \emph{bounded} transition system, whose states are of size at most $k$.

\subsection{Finite-Domain Semantics of a First-Order Transition System}

The finite-domain semantics of a first-order specification of a transition system  $\Tspec$  
is given by the transition system $\Tpar=(S_\para,\init_\para,R_\para)$ where $S_\para,\init_\para,R_\para$ are defined similarly to the corresponding definitions in $\Tfol$, 
except that only \emph{finite} domains are included.
Note that a semantics that \emph{bounds} the domain size by some $k \in \nat$ can be captured  with the standard first-order semantics by adding to $\Gamma$ the formula  $\mathrm{size}_{\leq k}:=\exists x_1,\ldots,x_k. \forall x. \vee_i (x = x_i)$ as an axiom. However,   the same cannot be done to capture the finite-domain semantics since finiteness of the domain is not definable in first-order logic.

Every trace of $\Tpar$ is also a trace of $\Tfol$. Hence, every safety property that is satisfied by $\Tfol$ is also satisfied by $\Tpar$, but the converse need not hold, as demonstrated by the following example.

\begin{example}[TreeTermination]\label{ex1}  
Consider a simple termination detection protocol for nodes arranged in a rooted tree (which captures the second phase of a simple broadcast protocol~\cite{distributed_network_protocols}). Upon receiving an $\mathsf{ack}$ from all its children, a node terminates and sends an $\mathsf{ack}$ to its parent, except for the root, which does not have a parent. The protocol satisfies the safety property that if the root has declared termination then all nodes have terminated as well.

\lstdefinestyle{myStyle}{
    belowcaptionskip=1\baselineskip,
    breaklines=true,
    frame=none,
    numbers=none,
    basicstyle=\scriptsize \ttfamily,
    keywordstyle=\bfseries\color{green!40!black},
    commentstyle=\itshape\color{purple!40!black},  
    mathescape = true,
    identifierstyle=\color{black},
    stepnumber=1,
    firstnumber=1,
    numbers=left,
    numberstyle=\tiny,
    escapechar=@,
}
\lstdefinelanguage{ivy}{
	keywords=[1]{module, sort, type, types, relation, individual, def, constant, safety, update, hint, condition, bound, definition, transition, invariant, function, axiom,
		assume, init, require},
	keywordstyle=[1]{\bfseries},
	keywords=[2]{true,false},
	keywordstyle=[2]{\bfseries\color{purple}},
	keywords=[3]{mutable,immutable},
	keywordstyle=[3]{\bfseries},
	comment=[l]{\#}
}

\begin{figure}[t]
\centering

\lstset{style=mystyle,linewidth=0.43\textwidth}

    \begin{lstlisting}[language=ivy,multicols=2]
sort node
immutable constant root : node
immutable relation leq(node,node)
axiom $\forall $X. leq(root,X) $\wedge$ leq(X,X) @\label{line:order1}@
axiom $\forall $X,Y. leq(X,Y) $\wedge$ leq(Y,X) $\to$ X=Y@\label{line:order2}@
axiom $\forall $X,Y,Z. leq(X,Y) $\wedge$ leq(Y,Z)$ \ \ \ \ \ \ \to$ leq(X,Z) @\label{line:order3}@
axiom $\forall $X,Y,Z. leq(Y,X) $\wedge$ leq(Z,X) $\ \ \ \ \ \ \to$ leq(Y,Z) $\vee$ leq(Z,Y) @\label{line:cycle_freeness}@
def child(Y,X) := leq(X,Y) $\wedge$ X$\neq$Y $\wedge$          $\ \forall$Z.leq(Z,Y)$ \ \wedge$ Z$\neq$Y $\to$leq(Z,X)
relation termd(node)
relation ack(node,node)
init $\neg$termd(X)  
init $\neg$ack(X,Y)

transition terminate(n: node)
  assume $\forall$X. child(X,n) $\to$ ack(X,n)
  $\forall$X. termd'(X) = termd(X) $\vee$ (X = n)
  $\forall$X,Y. ack'(X,Y) = ack(X,Y) $\vee$ child(n,Y)
  
safety termd(root) $\to$ $\forall$X. termd(X)


\end{lstlisting}

    \caption{
    First-order specification of the TreeTermination protocol.\label{tree_spec}
}
\end{figure}

A first-order specification of this protocol is given in \Cref{tree_spec}. The tree structure is %
modeled as proposed in~\cite{padon_phd}, by a 
relation $\mathsf{leq}$, axiomatized to induce a tree, with the intention that $\mathsf{leq(X,Y)}$ indicates that $\mathsf{Y}$ is a descendant of $\mathsf{X}$ in the tree. The axioms in \cref{line:order1,line:order2,line:order3} ensure that $\mathsf{leq}$ is a
partial order %
with the root as its minimum
and \cref{line:cycle_freeness} guarantees the tree structure. The formula $\mathsf{child(Y,X)}$ indicates that $\mathsf{Y}$ is a child of $\mathsf{X}$ by asserting that $\mathsf{X}$ is a predecessor of $\mathsf{Y}$ in the order. %

For finite instances, the specification satisifies safety: any node $X$ in the tree is connected to the root by a finite sequence of predecessor steps. It follows by induction on the length of this sequence that if the root terminated, so did $X$.
The induction step follows from the precondition of $\mathsf{terminate}$, which ensures that a node terminates only after its children sent it acks, 
and so 
already terminated.

While this proof is correct, it uses induction on the length of tree paths, and as a consequence, relies on the finiteness of the tree. In fact, if we allow infinite structures, the safety property may be violated. 
A counterexample trace for an infinite instance of the protocol 
is given in \Cref{inf_cex}.
In the example, the partial order is in fact a total order that consists of infinitely many nodes, all of which are larger than the root, but for none the root is a predecessor.
The counterexample consists of a single step, where the root, which has no children and therefore satisfies the pre-condition of the termination transition vacuously, terminates, even though all the other nodes did not terminate. 

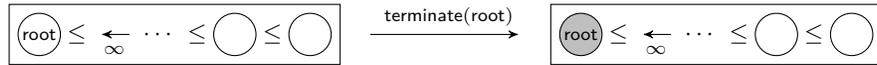
\begin{figure}[t]
\centering

\begin{tikzpicture}[->,>=stealth',shorten >=1pt,auto,scale=0.8]
		
		\tikzset{enode/.style={circle,draw,minimum size=16pt,inner sep=0mm,font=\scriptsize}}
		\tikzset{fun/.style={font=\tiny}}
		\tikzset{dots/.style={font=\small}}
		\tikzset{rel/.style={font=\small}}

            \draw[black] (2.5,2.5) rectangle (8,3.5);
            
		\node[enode]
		(e1)
		at (3,3)
		{ $\mathsf{root}$ };
            \node
		(leq1) at (3.625,3)
		{$\leq$};

            \node
		(leq2) at (5.625,3)
		{$\leq$};
            \node[enode]
		(e3) at (6.25,3)
            {};
            \node
		(leq3) at (6.875,3)
		{$\leq$};
            \node[enode]
		(e4) at (7.5,3)
            {};
		\node
		(e2)
            at (5,3)
            {$\cdots$};

            \node (einf) at (4.25,2.75) {\scriptsize $\infty$};

            \path [black] (4.5,3) edge  (4,3) {};

            \path (8.5,3) edge node { \scriptsize{ $\mathsf{terminate(root)}$} } (11,3);

            \draw[black] (11.5,2.5) rectangle (17,3.5);

\node[enode,fill=lightgray]
(e11)
		at (12,3)
		{ $\mathsf{root}$ };
            \node
		(leq11) at (12.625,3)
		{$\leq$};

            \node
		(leq21) at (14.625,3)
		{$\leq$};
            \node[enode]
		(e31) at (15.25,3)
            {};
            \node
		(leq31) at (15.875,3)
		{$\leq$};
            \node[enode]
		(e41) at (16.5,3)
            {};
		\node
		(e21)
            at (14,3)
            {$\cdots$};

            \node (einf1) at (13.25,2.75) {\scriptsize $\infty$};

            \path [black] (13.5,3) edge  (13,3) {};

	\end{tikzpicture}%

\caption{A counterexample to safety of TreeTermination over infinite structures, depicted as a trace of two states related by the transition relation: the first is an initial state, and the second violates safety. The order of nodes is depicted by $\leq$ and terminated nodes are depicted by shading. 
\label{inf_cex}} 
\end{figure}

In the sequel, we develop an approach that mimics induction over the state size and is able to verify the protocol for all finite instances by showing that every instance with more than $2$ nodes is simulated by an instance where some node other than the root is deleted.

\end{example}

\subsection{Establishing Cutoffs by Size-Reducing Simulations}

To verify first-order transition systems under the finite-domain semantics, we adapt the cutoff method to fit the first-order setting.
Fix a first order specification $\Tspec$, 
its corresponding transition systems $\{T_\domain\}_\domain$, its finite-domain semantics $\Tpar$, a safety specification $\safetyspec$, and the induced safety property $\safetysem$.
We define the notion of a cutoff 
for $\Tpar$ and $\safetysem$ as an adaptation of the standard definition~\cite{namjoshi2007}.

\begin{definition}\label{def:cutoff}
    A natural number $k$ is a \emph{cutoff} for $\Tpar$ and %
    $P$ if 
    \vspace{-0.2cm}
    \[
    \Tpar \models P  \qquad \Longleftrightarrow \qquad T_\domain \models P \text{ for all } \domain \text{ with } |\domain |\leq k.
    \]
\end{definition}

Generally speaking, cutoffs allow to reduce verification of parameterized systems to verification of bounded systems: if $k$ is a cutoff, then it suffices to verify  $T_\domain$ for domains $\domain$ with $|\domain| \leq k$ in order to deduce that $\Tpar \models P$.
Traditionally (e.g., in \cite{reasoning_about_rings}), cutoffs are investigated for \emph{classes} of transition systems and properties. 
Since \Cref{def:cutoff} considers cutoffs for a specific specification of a transition system and a specific safety property, a cutoff always exists: if $\Tpar\models P$, then $T_\domain\models P$ for any finite $\domain$, %
making any $k$ a cutoff. Similarly, if $\Tpar\not\models P$, then there exists a violating trace with some finite domain $\domain$, making $k = |\domain|$ a cutoff.
The challenge our work addresses, and what makes cutoffs useful for verification, is to establish a cutoff without checking whether $\Tpar \models P$.
To do so, we use the notion of a size-reducing simulation.

\begin{definition}\label{def:size_reducing_sim}
    A \emph{size-reducing simulation} for $ \Tpar$, $P$ and a bound $k$ is a relation $H\subseteq S_\para\times S_\para $ such that:
\begin{enumerate}
    \item Size Reduction: For every pair $(s^h,s^\ell)\in \simul$ we have $|s^\ell| < |s^h|$.  
    \label{item:size}
    \item Initiation: For every state $s^h\in \init_\para$ with $|s^h| > k$ there exists $s^\ell\in  \init_\para$ such that $(s^h,s^\ell)\in \simul$. \label{item:initation}
    \item Simulation: For every pair $(s^h,s^\ell)\in \simul$ and $s^{h\prime} \in S_\para$ such that $(s^h,s^{h\prime})\in R_\para$ there exists $s^{\ell\prime}\in S_\para$ such that $(s^{h\prime},s^{\ell\prime})\in \simul$ and $(s^\ell,s^{\ell\prime})\in R_\para^\star$. \label{item:simulation}
    \item Fault Preservation: For every pair $(s^h,s^\ell)\in \simul$, if $s^h\notin P$ then $s^\ell\notin P$. \label{item:bad}
\end{enumerate}
\end{definition}

\noindent Usually (e.g., in~\cite{reasoning_about_rings,namjoshi2007}) simulation relations for showing cutoffs are defined between different instances of the parameterized system. %
Instead, a size-reducing simulation is a self-simulation of $\Tpar$, which is more natural in this context.

\begin{theorem}\label{theorem:soundness}
    If $\simul$ is a size-reducing simulation for $\Tpar$, $P$  and $k$, then $k$ is a cutoff for $\Tpar$ and $P$.
\end{theorem}

Intuitively, the theorem holds because we can use $H$ to translate a trace violating $P$ in system instances of (finite) size larger than $k$ to a (perhaps shorter/longer) trace violating $P$ in a system instance with a smaller domain. Then, continuing by induction, we get a violating trace of size smaller than $k$. Thus if no counterexample to safety exists for system instances whose domain size is at most $k$, no counterexample exists for $\Tpar$, making $k$ a cutoff.

Our goal is to encode size-reducing simulations in first-order logic. 
As written, \cref{item:initation,item:simulation} involve alternation between universal and existential quantification over states. Since states are first-order structures, explicitly quantifying over them requires higher-order quantification. 
Note, however, that if a claim only involves  universal quantification over states, the quantification can remain implicit by checking (first-order) validity. 
We use this observation to strengthen \Cref{def:size_reducing_sim} to a \emph{strong} size-reducing simulation, 
by over-approximating the existential quantification using universal quantification together with a totality requirement that guarantees the universal quantification is not vacuous.

Namely, the existential quantification specifies the existence of a state $s^{\ell\prime}$ that 
simulates $s^{h\prime}$ and is a successor of $s^{\ell}$. When turning it into universal quantification, we can either quantify over all successor states and require that they are all simulating, or vice versa. We choose the latter option. 
We also strengthen the requirement $(s^\ell,s^\ellpr)\in R_\para^\star$ to $(s^\ell,s^\ellpr)\in R_\para$ or $s^\ell = s^\ellpr$, hence removing 
the existential quantification on intermediate states.
This is illustrated in \Cref{fig:simulation_diagram}.
The totality check then requires that a simulating state exists initially and that it continues to hold inductively. These still involve existential quantification over states, and we will see in the next section how to reduce them to first-order quantification for the specific class of size-reducing relations that we use.

\begin{definition}\label{def:strong_size_reducing_sim}    
    A \emph{strong size-reducing simulation} for  $\Tpar$, $P$ and $k$ is a relation $H\subseteq S_\para\times S_\para$ such that:
\begin{enumerate}
    \item Size Reduction: For every pair $(s^h,s^\ell)\in \simul$ we have $|s^\ell| < |s^h|$.  
    \label{item:strong_size}
    \item Strong Initiation:  For every pair $(s^h,s^\ell)\in \simul$, if $s^h\in \init_\para$ then $s^\ell \in \init_\para$. \label{item:strong_initation}
    \item Strong Simulation: For all pairs $(s^h,s^\ell)\in \simul$ and $(s^{h\prime},s^{\ell\prime})\in \simul$, if $(s^h,s^{h\prime})\in R_\para$ 
    and $\dom(s^\ell)=\dom(s^{\ell\prime})$
    then $(s^\ell,s^{\ell\prime})\in R_\para$ or $s^\ell=s^{\ell\prime}. $\label{item:strong_simulation}
    \item Fault Preservation: For every pair $(s^h,s^\ell)\in \simul$, if $s^h\notin P$ then $s^\ell\notin P$. \label{item:strong_bad}
    \item \label{item:strong_inductive_totality} Inductive Totality: \begin{itemize}
        \item For every $s^h \in \init_\para$ with $|s^h|>k$ there exists $s^\ell\in S_\para$ s.t.\ $(s^h,s^\ell)\in \simul$.
        \item For every pair $(s^h,s^{h\prime})\in R_\para$, if there exists $s^\ell$ s.t.\ $(s^h,s^\ell)\in H$ then there exists $s^{\ell\prime}\in S_\para$ s.t.\ $(s^{h\prime},s^{\ell\prime}) \in H $ and $\dom(s^\ell)=\dom(s^\ellpr)$.
    \end{itemize}
\end{enumerate}
\end{definition}

Note that in the strong simulation requirement we do not consider all states $s^\ellpr$ that simulate $s^\ell$ but only those that have $\dom(s^\ell)=\dom(s^\ellpr)$, since each trace is confined to one system instance. As a result, we also need to require the domain equality in the totality of \cref{item:strong_inductive_totality}.

\begin{figure}[t]
\usetikzlibrary{arrows.meta}

    \centering
\begin{tikzpicture}[->,>=stealth',shorten >=1pt,auto,scale=0.8]
		
		\tikzset{enode/.style={circle,draw,minimum size=18pt,inner sep=0mm}}

  		\tikzset{dnode/.style={circle,draw,dashed,minimum size=18pt,inner sep=0mm}}
		\tikzset{fun/.style={font}}
		\tikzset{dots/.style={font=}}
		\tikzset{rel/.style={font}}

        \node [enode]
		(e1)
		{ ${s^h}$ };

            \node
            (e0)
            [left = 0.5cm of e1]
            {\Cref{def:size_reducing_sim}};

            \node [enode]
		(e2)
		[right = 1.5cm of e1]
		{ $s^{h\prime}$ };

            \path (e1) edge node {$R$} (e2); 

            \node [enode]
		(e3)
            [below = of e1]
		{ ${s^\ell}$ };

            \node [dnode]
		(e4)
		[right = 1.5cm of e3]
		{ $s^{\ell\prime}$ };

            \node
		(exists)
		[right = 0.2pt of e4]
		{ $\exists$ };

            \path (e3) edge [dashed] node {$R^\star$} (e4); 
		\path (e1) edge node[swap] {$H$} (e3); 
  		\path (e2) edge [dashed] node {$H$} (e4); 

            \node [enode]
		(e5)
            [right = 6cm of e1]
		{ ${s^h}$ };

            \node
            (e0)
            [left = 0.5cm of e5]
            {\Cref{def:strong_size_reducing_sim}};

            \node [enode]
		(e6)
		[right = 1.5cm of e5]
		{ $s^{h\prime}$ };

            \path (e5) edge node {$R$} (e6); 

            \node [enode]
		(e7)
            [below = of e5]
		{ ${s^\ell}$ };

            \node [enode]
		(e8)
		[right = 1.5cm of e7]
		{ $s^{\ell\prime}$ };

            \path (e7) edge [dashed] node {$R^=$} (e8); 
		\path (e5) edge node[swap] {$H$} (e7); 
  		\path (e6) edge node {$H$} (e8); 

            \draw[-Implies,line width=0.5pt,double distance=2.5pt] (4,-1) -- (5.5,-1);

 \end{tikzpicture}%
    \caption{Strengthening of simulation to strong simulation: solid nodes are universally quantified while the dashed node is existentially quantified; solid edges represent assumptions and dashed edges represent conclusions. $R^=$ is a shorthand for $R 
    \cup \mathrm{id}_S
    $.}
    \label{fig:simulation_diagram}
\end{figure}
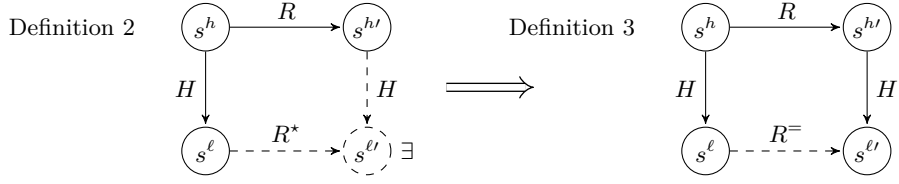

\begin{lemma}\label{lem:strong_is_weak}
    A strong size-reducing simulation is a size-reducing simulation. 
\end{lemma}

The converse of \Cref{lem:strong_is_weak} does not hold. In fact,  a transition system may have a size-reducing simulation but no strong size-reducing simulation.   
Moreover, while size-reducing simulations are as powerful as inductive invariants as a proof rule~\cite{namjoshi2007}, the proof does not carry over to strong size-reducing simulations.

\section{First-Order Encoding of Size-Reducing Simulations}\label{sec:smt_encoding}

In this section we show how to encode and validate a strong size-reducing simulation in first-order logic. 
To circumvent the need to express and verify size reduction in first-order logic, 
we consider a certain class of strong size-reducing simulations, where the pairs  $(s^h,s^\ell)$ in the relation are such that $s^\ell$ is obtained from $s^h$ by deleting an element in its domain and updating the interpretation of all symbols. %
We propose a first-order encoding of such relations and show how to validate that the relation defined by this encoding is a strong size-reducing simulation using a set of verification conditions expressible in first-order logic.

For the remainder of the section, fix 
a first-order specification $\Tspec=(\Sigma,\Gamma,\iota,\tau)$, a safety specification $\safetyspec$ and a natural number $k$. Let $\Tpar$ be the finite-domain semantics associated with $\Tspec$ and $P$ the safety property associated with $\safetyspec$.
    
\subsection{Encoding Size-Reducing Relations Using Substructures}

A strong size-reducing simulation is a relation between states that satisfies certain requirements. 
To specify a relation between states 
(structures) we rely on the standard approach of using a formula over 
a double vocabulary $\Sigma^h \uplus \Sigma^\ell$, where the interpretation of $\Sigma^h$ represents the ``high'' state $s^h$, and the interpretation of $\Sigma^\ell$ represents the ``low'' state $s^\ell$. However, since a model of a formula (even when it is  a two-vocabulary formula) has a single domain,  such an encoding does not account for size reduction of the domain. 
Our observation is that we can enforce size reduction by deleting an element from the  domain.
Thus, we use a formula $\sq(z)$ over  $\Sigma^h \uplus \Sigma^\ell$ to specify a size-reducing relation, where the deleted element is specified using the free variable $z$.
In a model of $\sq(z)$, the interpretation of $\Sigma^h$ represents $s^h$ and the state $s^\ell$ is a substructure of the interpretation of $\Sigma^\ell$ with the same domain elements as $s^h$, excluding the element assigned to $z$.
We formalize this by the following definition.

\begin{definition}\label{def: relation_from_sq}
Let $\sq(z)$ be a formula over 
$\Sigma^h \uplus \Sigma^\ell$, $(\domain, \interp^h \uplus \interp^\ell)$ a structure over $\Sigma^h \uplus \Sigma^\ell$ and $d_0 \in \domain$ such that $(\domain,\interp^h\uplus\interp^\ell),[z\mapsto d_0]\models \sq(z)$. Then for $\exc\domain = \domain \setminus\{d_0\}$, if there exists an interpretation $\exc {\interp^\ell}$ such that $(\exc\domain,\exc {\interp^\ell})$ is a substructure of $(\domain,\interp^\ell)$, we say that $s^\ell=(\exc\domain,\exc {\interp^\ell})$ is an \emph{$\sq$-substructure} of $s^h=(\domain,\interp^h)$.
The \emph{size-reducing relation} of $\sq$ is the relation 
\[
\simul_\sq = \left\{(s^h,s^\ell) \in  \structset(\Sigma)\times \structset(\Sigma) \mid 
s^\ell \text{ is an } \sq\text{-substructure of  }s^h \right\}.
\]

\end{definition}

\noindent 
An example of a pair $(s^h,s^\ell)\in H_\eta$ is depicted in \Cref{fig:substructure} along with the structure $(\domain, \interp^h \uplus \interp^\ell)$ and the element $d_0\in \domain$ that induces them, assuming that $(\domain, \interp^h \uplus \interp^\ell),[z \mapsto d_0] \models \sq(z)$.
In the example: $\domain = \{1,2,3,4\}$, $d_0=3$, 
$\interp^h= \left\{ 
\mathsf{leq} :  \{ (1,2),(1,3),(3,4),(1,4) \},
\mathsf{root} : 1,
\mathsf{ack} : \{(4,3)\},
\mathsf{termd} : \{4\}
\right\}$
and 
$\interp^\ell = \left\{ 
\mathsf{leq} :  \{ (1,2),(1,3),(3,4),(1,4) \},
\mathsf{root} : 1,
\mathsf{ack} : \{(4,1)\},
\mathsf{termd} : \{4\}
\right\}$.
This induces the pair $(s^h,s^\ell)\in \simul_\sq$ where $s^h=(\domain,\interp^h)$, $s^\ell = (\exc\domain,\exc\interp^\ell)$, $\exc \domain = \{1,2,4\}$ and
$\exc\interp^\ell = \left\{ 
\mathsf{leq} :  \{ (1,2),(1,4) \},
\mathsf{root} : 1,
\mathsf{ack} : \{(4,1)\},
\mathsf{termd} : \{4\}
\right\}$. The choice $d_0=4$ would give a different pair in $\simul_
\sq$.

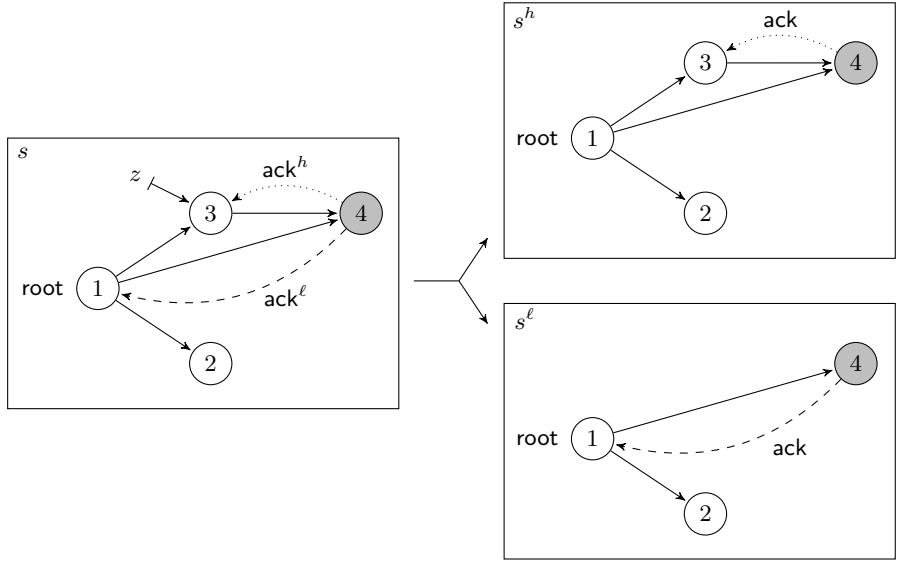
\begin{figure}[ht]
    \usetikzlibrary{arrows.meta}

    \centering

\begin{tikzpicture}[->,>=stealth',shorten >=1pt,auto,scale=0.8]
		
\tikzset{enode/.style={circle,draw,minimum size=16pt,inner sep=0mm}}
		\tikzset{fun/.style={font}}
		\tikzset{dots/.style={font=}}
		\tikzset{rel/.style={font}}

    \node [enode] (e1) { $1$ };
    \node (s) [above of =  e1,xshift=-1cm,yshift=0.8cm] { $s$ };
    \node (root) [left = 0.05cm of e1] {$\mathsf{root}$};
    \node [enode] (e2) [below of=e1,xshift=1.5cm] { $2$ };
    \node [enode] (e3) [above of=e1,xshift=1.5cm] { $3$ };
    \node [enode,fill=lightgray] (e4) [right of=e3, xshift=1cm] { $4$ };
    \node (ackl) at (3.125,-0.125) {$\mathsf{ack}^\ell$};

    \path (e1) edge node[swap] {} (e2);
    \path (e1) edge node {} (e3);
    \path (e3) edge node {} (e4);
    \path (e1) edge node[swap] {} (e4);

    \path (e4) edge [dotted,bend right] node [anchor=south] {$\mathsf{ack}^h$} (e3);
    \path (e4) edge [dashed,bend left] node {} (e1);

    \node (e5) [left of = e3, yshift = 0.5cm] { $z$ };
    \path (e5) edge [|->] (e3);

    \node [enode] (e1high) [right =6cm of e1, yshift = 2cm] { $1$ };
    \node (shigh) [above of =  e1high,xshift=-0.9cm,yshift=0.6cm] { $s^h$ };
    \node [enode] (e2high) [below of=e1high,xshift=1.5cm] { $2$ };
    \node [enode] (e3high) [above of=e1high,xshift=1.5cm] { $3$ };
    \node [enode,fill=lightgray] (e4high) [right of=e3high, xshift=1cm] { $4$ };
    \node (root) [left = 0.05cm of e1high] {$\mathsf{root}$};

    \path (e1high) edge node[swap] {} (e2high);
    \path (e1high) edge node {} (e3high);
    \path (e3high) edge node {} (e4high);
    \path (e1high) edge node[swap] {} (e4high);

    \path (e4high) edge [dotted,bend right] node [anchor=south] {$\mathsf{ack}$} (e3high);

    \node [enode] (e1low) [right =6cm of e1, yshift = -2cm] { $1$ };
    \node (slow) [above of =  e1low,xshift=-0.9cm,yshift=0.6cm]  { $s^\ell$ };

    \node [enode] (e2low) [below of=e1low,xshift=1.5cm] { $2$ };
    \node (e3low) [above of=e1low,xshift=1.5cm] {};
    \node [enode,fill=lightgray] (e4low) [right of=e3low, xshift=1cm] { $4$ };
    \node (root) [left = 0.05cm of e1low] {$\mathsf{root}$};

    \path (e1low) edge node[swap] {} (e2low);
    \path (e1low) edge node[swap] {} (e4low);

    \path (e4low) edge [dashed,bend left] node {} (e1low);
    \node (ackl) at (11.5,-2.625) {$\mathsf{ack}$};

    \draw[black] (-1.5,-2) rectangle (5,2.5);
    \draw[black] (6.75,0.5) rectangle (13.25,4.75);
    \draw[black] (6.75,-4.5) rectangle (13.25,-0.25);

    \draw[black] (5.25,0.125) rectangle (6,0.125);
    \draw [->] (6,0.125) -- (6.5,0.875);
    \draw [->] (6,0.125) -- (6.5,-0.625);

 \end{tikzpicture}%
     \caption{
    Left: a structure $s=(\domain, \interp^h \uplus \interp^\ell)$ and an assignment $[z \mapsto d_0]$ that satisfy $\sq(z)$. The $h$ and $\ell$ copies of $\mathsf{root}$,  $\mathsf{leq}$ and $\mathsf{termd}$ are interpreted in the same way,  so these interpretations are depicted jointly. 
    $\mathsf{leq}(\cdot,\cdot)$ is depicted by solid lines and $\mathsf{termd}(\cdot)$ is depicted by shading. The interpretation of the two copies of  $\mathsf{ack}(\cdot,\cdot)$ are depicted by dashed and dotted arrows.
    Right: the derived high and low structures: $s^h=(\domain,\interp^h)$ and $s^\ell=(\exc \domain, \exc\interp^\ell)$, with the relations and constants depicted in the same way. 
        \label{fig:substructure}
    }
\end{figure}

\begin{lemma}\label{lem:size_reduction}
    For a formula $\sq(z)$ over $\Sigma^h\uplus\Sigma^\ell$, and a pair $(s^h,s^\ell)\in \simul_\sq$ such that $s^h$ is finite, we have $|s^\ell| < |s^h|$.
\end{lemma}

\Cref{lem:size_reduction} states that, when  restricting to finite structures,  $\simul_\sq$ is size-reducing. This holds by construction.
In the sequel, we define a set of verification conditions that guarantee the rest of the required properties of \Cref{def:strong_size_reducing_sim}.

In order to encode the verification conditions, we need to be able to reason about the states  $(s^h,s^\ell) \in \simul_\sq$.  
The state $s^h$ is captured by $\sq(z)$ directly. However, 
$s^\ell$ is defined by a substructure that excludes the element assigned to $z$ from the domain of $\sq(z)$. As a result, we need to exclude $z$ when reasoning about $s^\ell$. 
For this purpose, we define, for a formula $\formula$ and a variable $z\notin \Vars(\formula)$, the $z$-excluded representation, denoted $\ADrep z \formula $.
The $z$-excluded representation is
the formula obtained from $\formula$ by adding to each quantifier a guard that excludes $z$.
This ensures that the evaluation of $\ADrep z \formula $ on a structure with domain $\domain$  with $z$ assigned to $d_0\in \domain$ is the same as the evaluation of $\formula$ on a substructure with domain $\domain \setminus \{d_0\}$ (if such a substructure exists, see \Cref{subsec:guaranteeing_totality}).

\begin{definition}
    For a formula $\formula$ and a variable $z\notin \Vars(\formula)$, the $z$-excluded representation of $\formula$, denoted $\ADrep{z}{\formula}$, is a formula defined inductively as follows:
    \begin{itemize}
        \item $\ADrep{z}{\formula}=\formula $ if $\formula=r(t_1,\ldots,t_m)$ or $\formula = (t_1 = t_2)$
        \item $\ADrep{z}{\neg \formulaa} = \neg \ADrep{z}{\formulaa}$
        \item $\ADrep z {\formulaa_1 \bowtie \formulaa_2 } = \ADrep{z}{\formulaa_1} \bowtie \ADrep{z}{\formulaa_2} $ for $\bowtie\ \in \{\wedge,\vee,\to \} $
        \item $\ADrep{z}{\forall x. \formulaa} = \forall x. (x\neq z) \to \ADrep{z}{\formulaa} $
        \item $\ADrep{z}{\exists x. \formulaa} = \exists x. (x\neq z) \wedge \ADrep{z}{\formulaa} $
    \end{itemize}
\end{definition}

The following lemma formalizes the  aforementioned relation between 
the value of the $z$-excluded representation of $\formula$ in a structure
and 
the value of $\formula$ in its substructure that excludes the value of $z$.

\begin{lemma}\label{lem:domain_extension}
Let $\formula$ be a closed formula over $\Sigma$ and $z\notin \Vars(\formula)$ a variable.
Let $(\domain,\interp)$ be a structure for $\Sigma$, $d_0\in \domain$ and $(\exc\domain,\exc\interp)$  a substructure of $(\domain,\interp)$ such that $\exc \domain = \domain \setminus \{d_0\} $. Then %
    $
    (\exc \domain,\exc \interp) \models \formula \iff (\domain, \interp),[z\mapsto d_0] \models \ADrep{z}\formula.
    $
\end{lemma}

\subsection{Validating Simulation}

Given a formula $\sq(z)$ over $\Sigma^h\uplus\Sigma^\ell$, we now go on to define a set of verification conditions for $\Tspec$, $\safetyspec$, $k$
and $\sq(z)$ that ensure that
$\simul_\sq\cap (S_\para \times S_\para)$ is a strong size-reducing simulation for $\Tpar$, $\safetysem$ and $k$. 
The full set of verification conditions is listed in \Cref{fig:verificationsconditions}.
In this subsection we introduce three of them, 
which guarantee %
strong initiation, strong simulation and fault preservation (\cref{item:strong_initation,item:strong_simulation,item:strong_bad} of \Cref{def:strong_size_reducing_sim}).
Having already ensured size reduction (\cref{item:strong_size}) in \Cref{lem:size_reduction},  
the next subsection addresses the final requirement: inductive totality (\cref{item:strong_inductive_totality}).

In all verification conditions, we use the appropriate copies of the vocabulary to refer to high states $(\Sigma^h)$ and low states $(\Sigma^\ell)$ and to pre-states (unprimed) and post-states (primed). We use $\sq(z)$ to encode pairs of states related by simulation and the $z$-excluded representation to encode properties of the low state in a pair.
For this purpose, we assume w.l.o.g.\ %
that the variable $z$ does not appear in the system and safety specification $\Gamma,\iota,\tau,\safetyspec$ (this can be achieved by renaming). 

Verification conditions (\ref{item:vc-iota}) and (\ref{item:vc-safety}) are  straight forward encodings of strong initiation and fault preservation, respectively.

Verification condition (\ref{item:vc-tau}) establishes strong simulation. 
It requires that any transition between high-states is reflected by a transition between their simulating low states or by a stutter, provided that the low states have the same domain.
The assumption that the %
high and low states are related by the simulation relation is encoded by asserting $\sq(z)$ (for the pre-states) and $\sq'(z)$ (for the post-states).
Crucially, we use the same variable $z$ to ensure that we refer to low states with the same domain. 
We encode a stutter by the formula $\idle$ given by
$\idle := \bigwedge_{r\in \Sigma} \left( \forall \seq x.\ r'(\seq x) \leftrightarrow r(\seq x) \right) \wedge \bigwedge_{f\in \Sigma} (\forall \seq x.\ f'(\seq x) = f(\seq x)).$

\begin{lemma} \label{lem:vc_simulation}
Let $\sq(z)$ be a formula over $\Sigma^h\uplus \Sigma^\ell$. If verification conditions (\ref{item:vc-iota})--(\ref{item:vc-safety}) 
of \Cref{fig:verificationsconditions} are valid then $\simul_{\sq}\cap (S_\para \times S_\para)$ satisfies \cref{item:strong_initation,item:strong_simulation,item:strong_bad} of \Cref{def:strong_size_reducing_sim}.
\end{lemma}

\subsection{Guaranteeing and Validating Totality}\label{subsec:guaranteeing_totality}

\begin{figure}[t]
\newcounter{vccounter}
\setcounter{vccounter}{0}
    \centering
\begin{tabular}{lll}  
\refstepcounter{vccounter}(\thevccounter{}\label{item:vc-iota}) \hspace{0.5cm} & 
$\Gamma^h \wedge \iota^h \wedge \sq(z)  \implies  \ADrep{z}{\iota}^\ell$ \hspace{3cm}  & 
      [$\iota$-preservation] \\
\refstepcounter{vccounter}(\thevccounter{}\label{item:vc-tau}) & 
$\Gamma^h \wedge \Gamma^{h\prime} \wedge \sq(z)  \wedge \sq'(z) \wedge \tau^h \implies \ADrep{z}{\tau\vee \idle}^\ell $  &
[$\tau$-preservation] \\
\refstepcounter{vccounter}(\thevccounter{}\label{item:vc-safety}) & 
$ \Gamma^h \wedge \neg \safetyspec^h \wedge \sq(z) \implies \ADrep{z}{\neg \safetyspec}^\ell$ & [$\neg\safetyspec$-preservation] \\

\refstepcounter{vccounter}(\thevccounter{}\label{item:vc-project}) & 
$\Gamma^h \wedge \sq(z) \implies \closure z ^\ell $ & [projectability] \\

\refstepcounter{vccounter}(\thevccounter{}\label{item:vc-gamma}) & 
$\Gamma^h \wedge \sq(z) \implies \ADrep {z }{\Gamma}^\ell $   &
[$\Gamma$-preservation] \\

\refstepcounter{vccounter}(\thevccounter{}\label{item:vc-muinit}) & 
$\Gamma^h\wedge \iota^h \wedge \size \implies \exists z. \cond(z) $ & [$\cond$-initation] \\
\refstepcounter{vccounter}(\thevccounter{}\label{item:vc-muconsec}) & 
$\Gamma^h \wedge \Gamma^{h\prime} \wedge \cond  (z) \wedge \tau^h \implies \cond'(z) $  &
[$\cond$-consecution] \\

\end{tabular}
    \caption{Verification conditions establishing strong size-reducing simulation for a high-low update $\sq(z)$ with precondition $\cond(z)$, $\Tspec=(\Sigma,\Gamma,\iota,\tau)$, $\safetyspec$ and $k$.\label{fig:verificationsconditions}
 }
\end{figure}

So far, we have considered an arbitrary formula $\sq(z)$. 
We now define a syntactic restriction on $\sq(z)$, that, along with four verification conditions,  guarantees that the induced size-reducing relation $\simul_\sq \cap (S_\para \times S_\para)$ is inductively total (\cref{item:strong_inductive_totality} of \Cref{def:strong_size_reducing_sim}).
Namely, we consider formulas $\sq(z)$ that consist of a precondition $\cond(z)$ and update formulas for all the symbols in the vocabulary. For such formulas, called high-low updates, we show how to encode the remaining existential quantification over states in the definition of inductive totality by first-order quantification.

\begin{definition}
    \label{def:mu-update}
    A formula $\eta(z)$ over the vocabulary $\Sigma^h \uplus \Sigma^\ell$ is a \emph{high-low update}    if it has the form:
    $$ \sq(z) = \cond(z) \wedge \bigwedge_{r\in \Sigma} \left(\forall \seq x.\ r^\ell(\seq x) \leftrightarrow \formula_r(\seq x,z) \right) \wedge \bigwedge_{f\in \Sigma} \left(\forall \seq x.\ f^\ell(\seq x) = t_f(\seq x,z)\right)$$
    where $\cond(z)$ is a formula over $\Sigma^h$, for every relation symbol $r\in \Sigma$ of arity $m$,  $\formula_r$ is a formula over $\Sigma^h$ with $m + 1$ free variables, and for every function symbol $f\in \Sigma$ with arity $m$, $t_f$ is a term over $\Sigma^h$ with $m+1$ free variables.
    We call $\cond$ the \emph{precondition} for $\sq$, each $\formula_r$ the \emph{update formula} for $r$ and each $t_f$ the \emph{update term} for $f$. 
\end{definition}

Intuitively, %
for a high-low update $\sq(z)$ with precondition $\cond(z)$, 
to any high state and assignment to $z$ that satisfies the precondition $\cond(z)$, $\sq(z)$ defines the low state as a deterministic update from the high state, according to the formulas $\formula_r$ and terms $t_f$.
For example, we can use $\cond(z)$ to specify  that the deleted element must be different from all constants, and we can update all relation and function symbols to their values in the high state (we refer to this as a trivial update).
However, since the low state is obtained by deleting an element, we still need to make sure that, after the deletion, the updates result in a well-defined substructure, and,  furthermore, that the structure is in fact a state (satisfies $\Gamma$).
This is done by additional verification conditions.
First, we formalize the precise guarantee of a high-low update:

\begin{lemma}\label{lem:totality}
    Let $\sq(z)$ be a high-low update with precondition $\theta(z)$ and $(\domain,\interp^h)$  a structure with $d_0\in \domain$ such that $(\domain,\interp^h),[z\mapsto d_0]\models \cond(z)$. Then there exists an interpretation $\interp^\ell$ such that 
    $(\domain,\interp^h\uplus\interp^\ell),[z\mapsto d_0] \models \sq(z)$.\end{lemma}

The lemma states %
that whenever the high-state satisfies $\cond(z)$ for some value of $z$, there exists $\interp^\ell$ such that $\sq(z)$ holds. 
By \Cref{def: relation_from_sq}, to obtain the simulating state of the high-state,
we then take a substructure which excludes the assigned value of $z$. 
To ensure that a substructure can be obtained in this way,  
we must verify that for any function symbol $f$, 
its interpretation remains a total function
when we restrict its domain and range
by discarding the assigned value of $z$. To ensure that a state is obtained, we further need to verify that the resulting substructure satisfies the axioms.  These properties are guaranteed by verification conditions (\ref{item:vc-project}) and (\ref{item:vc-gamma}), where 
$\closure z  := \bigwedge_{f\in \Sigma} \closed z f $ and $\closed{z}{f} := \forall \seq x.\ 
(\wedge_{i=1}^m x_i \neq z)\to f(\seq x) \neq z$.

If the aforementioned verification conditions hold, we know that any high state that satisfies $\cond(z)$ has a corresponding low state  in $\simul_\sq \cap (S_\para\times S_\para)$.
To guarantee  inductively totality it remains to ensure that the precondition $\cond(z)$ holds over appropriate high states.  
Verification condition (\ref{item:vc-muinit}) requires 
that for all initial high states with size larger than $k$ there exists $z$ such that $\cond(z)$ holds (note the existential quantification). The size constraint is encoded by $\size := \exists x_1,\ldots,x_{k+1}. \wedge_{i,j} x_i \neq x_j$.
Verification condition (\ref{item:vc-muconsec}) ensures that totality is maintained over transitions between high states by 
requiring that $\cond(z)$ is maintained,
i.e., if $\cond(z)$ holds, then so does $\cond'(z)$. 
To ensure that the low states corresponding to the high pre-state and the high post-state have the same domain, we use the same variable $z$ in $\cond(z)$ and $\cond'(z)$. This ensures that we can delete the same element that is deleted in the pre-state in the post state, leading to the same domain.

\begin{lemma}\label{lem:vc_inductive_totality}
    Let $\sq(z)$ be a high-low update with precondition $\theta(z)$, if verification conditions (\ref{item:vc-project})--(\ref{item:vc-muconsec}) of 
    \Cref{fig:verificationsconditions} are valid,  
    then $\simul_\sq \cap (S_\para \times S_\para)$ satisfies \cref{item:strong_inductive_totality} of \Cref{def:strong_size_reducing_sim}.
\end{lemma}

\begin{remark}\label{remark:invariants}
The main role of $\cond(z)$ is to specify the deleted element in a high state in a simulation pair. However, it can also be used to restrict the high states, which may be necessary to make verification conditions (\ref{item:vc-tau})--(\ref{item:vc-gamma}) valid. This can be done by conjoining to the deletion condition in $\cond(z)$ a closed formula over $\Sigma^h$ that serves as an invariant. In this case, the inductive totality checks verify that the invariant is inductive, ensuring soundness of the restriction.
\end{remark}

\subsection{Putting It All Together}

The following theorem summarizes the soundness of the first-order encoding of a strong size-reducing simulation by a high-low update.

\begin{theorem}\label{theo:summary}
If all the verification conditions for a high-low update $\sq(z)$, $\Tspec$, $\safetyspec$ and $k$ 
(\Cref{fig:verificationsconditions}) are valid, then 
$\simul_\sq\cap (S_\para\times S_\para)$ is a strong size-reducing simulation for $\Tpar$, $P$ and $k$. Consequently, $k$ is a cutoff for $\Tpar$ and $P$.
\end{theorem}

\paragraph{Back to \Cref{ex1}.} %
\Cref{fig:sq_for_tree} encodes a high-low update that induces a size-reducing simulation for establishing a cutoff of $2$ for the TreeTermination protocol (\Cref{ex1}). 
In this example, violation of safety is witnessed by a node that did not terminate even though the root terminated. To allow the high-low update to refer to the node inspected by the property, we first Skolemize the negated safety property into 
$\mathsf{termd(root)\wedge \neg termd(n_{sk})}$, where $\mathsf{n_{sk}}$ is a fresh immutable constant that is added to the first-order specification of the transition system.
The original transition system violates the universally quantified safety property if and only if the transition system augmented with $\mathsf{n_{sk}}$ violates the Skolemized property.
(A similar Skolemization is performed in all of our examples.)
The precondition of the high-low update is then: $\cond(z)  = (z \neq \mathsf{n_{sk}}^h \wedge z \neq \mathsf{root}^h)$, i.e., the deleted element $z$ is required to be different than both the root and the violating node $\mathsf{n_{sk}}$, which ensures fault preservation.
The update of the $\mathsf{ack}$ relation is specified by the formula  $\formula_{\mathsf{ack}}$, which says that messages bypass the deleted node.
The other updates are trivial and thus not specified in \Cref{fig:sq_for_tree}. $\formula_{\mathsf{ack}}$ and the resulting high-low update are given below. 
\begin{align*}
    \formula_{\mathsf{ack}}(x,y,z) \ & =  ( \mathsf{ack}^h(x,y) \wedge x\neq z \wedge y \neq z) \\
   & \vee    (\mathsf{ack}^h(x,z) \wedge \mathsf{child}^h(z,y))
    \vee  (\mathsf{child}^h(x,z) \wedge \mathsf{ack}^h(z,y))\\
    \sq(z) = \cond(z) &
        \wedge (\forall x,y. \ \mathsf{leq}^\ell(x,y)\leftrightarrow\mathsf{leq}^h(x,y)) 
        \wedge (\forall x. \ \mathsf{termd}^\ell(x)\leftrightarrow\mathsf{termd}^h(x)) \\
        & \wedge  (\forall x,y. \ \mathsf{ack}^\ell(x,y)\leftrightarrow\formula_{\mathsf{ack}}(x,y,z))
        \wedge  (\mathsf{root}^\ell = \mathsf{root}^h) \wedge (\mathsf{n_{sk}}^\ell = \mathsf{n_{sk}}^h)
\end{align*}

The structure in the left of \Cref{fig:substructure} is an example of a model for $\sq(z)$ (the Skolem constant, which is omitted from \Cref{fig:substructure}, can be interpreted to any node other than node $3$).
With this definition of a high-low update, all the verification conditions are valid.
To see that $\tau$-preservation holds, consider a transition $\mathsf{terminate}(n)$ in the high state. 
If $n$ is equal to the deleted node $z$, all changes made by the transition are lost in the projection, so the transition is simulated by a stutter in the low state. 
If $n \neq z$, the transition is simulated by  $\mathsf{terminate}(n)$ in the low states: 
if $z$ is a child of $n$, then the updates to $\mathsf{ack}$ guarantee that the precondition for the transition still holds in the low state; 
if $z$ is the parent of $n$, 
the updates to $\mathsf{ack}$ guarantee that the updated post-state is indeed the post-state of the transition; finally, if $z$ is not a child nor a parent of $n$, the simulation is clear. 
We emphasize that when using our approach, validating the verification conditions for the given high-low update is delegated to an SMT solver (see  \Cref{subsec:implementation} for explanation of the hint in \Cref{fig:sq_for_tree}).

\begin{figure}[t]
    \centering
\lstset{style=mystyle}
\begin{lstlisting}[language = ivy]
bound node 2      # default value 
condition(z: node) =  z $\neq$ $\mathsf{n_{sk}}$ $\wedge$  z $\neq$ root            # default formula 
update ack(x: node, y: node, z: node) = 
    (ack(x,y) $\wedge$ x $\neq$ z $\wedge$ y $\neq$ z) $\vee$ (ack(x,z) $\wedge$ child(z,y)) $\vee$ (child(x,z) $\wedge$ ack(z,y))
hint terminate($\mathsf{n_h}$: node, $\mathsf{n_\ell}$: node, z: node) =  ($\mathsf{n_\ell}$ = $\mathsf{n_h}$)
  \end{lstlisting}

    \caption{Encoding of a high-low update for \Cref{ex1}.}
    \label{fig:sq_for_tree}
\end{figure}

\subsection{Extension to a Many-Sorted Vocabulary} \label{subsec:many_sorted}
We considered a single sort in the vocabulary used to specify $\Tspec$, which the finite-domain semantics treats as finite but unbounded. 
Our approach generalizes to transition systems specified using many-sorted first-order logic, including interpreted sorts and symbols. 
In this case, the finite-domain semantics affects only the sorts that are designated as finite.  The finite sorts and corresponding symbols are required to be uninterpreted, which means that they may only be restricted by a finitely-axiomatizable background theory (e.g., the theory of rooted trees as in \Cref{ex1}). 
Other sorts and symbols may be interpreted, i.e., defined by arbitrary background theories, including ones that are not  finitely-axiomatizable (e.g., linear integer arithmetic).

To handle multiple finite sorts, we apply our approach iteratively, each time establishing a cutoff for a single finite sort $\sort{s}$
by finding a size-reducing simulation for that sort.
When proving a cutoff for $\sort{s}$, we define the cardinality of a structure as the cardinality of $\domain(\sort s)$, and replace $\domain$ with $\domain(\sort s)$ in the definition of a cutoff and size-reducing simulation.
The variable $z$ used as a free variable in a high-low update is of sort $\sort{s}$, and 
in the definition of $\simul_\sq$ we define %
$\exc \domain (\sort s) = \domain (\sort s)\setminus \{d_0\}$. 
Finally, in the definition of a $z$-excluded representation and $\closed z f$ we add the guard $x\neq z$ only for variables $x$ of the sort $\sort s$.
Then, when moving on to establish a cutoff for the next sort
we restrict the cardinality of $\sort s$ to the established cutoff bound using the axiom $\mathrm{size}_{\leq k}$. 
In some cases, the order in which we apply our method matters. 
Specifically, when we have a function $f$ from finite sort $\sort {s_1}$ to finite sort $\sort {s_2}$, it is natural to first show a cutoff for $\sort{s_1}$.

\section{Implementation and Evaluation}\label{sec:evaluation}
\subsection{Implementation}\label{subsec:implementation}

The implementation of our method, which is available at \cite{artifact}, is built on top of the open-source verification tool {\mypyvy}~\cite{mypyvy}. 
{\mypyvy} allows a user to provide a first-order specification of a transition system $\Tspec$ and a safety property $\safetyspec$.
To run our tool, the user additionally provides a cutoff bound for one of the system's sorts, as well as a high-low update specified as a deletion condition for $\cond(z)$ 
along with update formulas and terms for relation and function symbols (\Cref{def:mu-update}). When omitted, 
the deletion condition is assigned the default formula $\bigwedge_{c} z\neq c$ where $c$ ranges over the immutable constants, including Skolem constants,
and the cutoff bound is assigned  the number of immutable constants.
Similarly, for an update term the default is $\formula_r(\seq{x_h},z)=r^h(\seq{x_h})$ and for an update formula the default is $t_f(\seq{x_h},z)=f^h(\seq{x_h})$.
Invariants of the system can also be provided, and are conjoined to $\cond(z)$ as described in \Cref{remark:invariants}.
For instance, the user input for \Cref{ex1} is given in \Cref{fig:sq_for_tree}. 
Our tool generates the verification conditions described in \Cref{fig:verificationsconditions} and validates them using the SMT solver Z3~\cite{z3}. 
If Z3 establishes that some verification condition is invalid, it generates a counterexample, which is a model of the negated verification condition. In such cases, the counterexample can point the user towards finding an appropriate fix for the high-low update.

\paragraph{Optimizations.}
Of the verification conditions in \Cref{fig:verificationsconditions}, item (\ref{item:vc-tau}) is the most challenging for the solver. This is because $\tau$ often has an $\exists^*$ prefix (quantifying over transition arguments) followed by nested quantifiers, as can be seen in \Cref{ex1}. Since $\tau$ appears on both sides of the implication in  verification condition (\ref{item:vc-tau}), quantifier alternations are introduced, which are difficult for SMT solvers to handle. 
To expedite this check, we split it into a separate check for each transition (disjunct) of $\tau$, strengthening the simulation 
requirement. Furthermore, we let the user provide hints for the  existential quantifiers (arguments) on the right-hand side of the implication, as described next.
Generally speaking,  the queries of verification condition (\ref{item:vc-tau}) have the form
$(\exists \seq{x_h}.\formula(\seq{x_h},z)) \implies (\exists \seq{x_\ell}. \formulaa(\seq{x_\ell},z))$,
where $\formula$ and $\formulaa$ are potentially complicated formulas with nested quantifiers.
We propose verifying the above using a user-provided \emph{hint} $\hint(\seq{x_h},\seq{x_\ell},z)$ that satisifies

\begin{compactitem}
    \item \textit{totality:} $\formula(\seq{x_h},z)\implies \exists \seq{x_\ell}. \hint(\seq{x_h},\seq{x_\ell},z)$,
    \item \textit{sufficiency:} $\formula(\seq{x_h},z) \wedge \hint(\seq{x_h},\seq{x_\ell},z)\implies \formulaa(\seq{x_\ell},z)$.
\end{compactitem}
In this way we break the query into two, eliminating the existential quantification over $\formulaa(\seq{x_\ell},z)$, and introducing one over $\hint(\seq{x_h},\seq{x_\ell},z)$.
The idea is for 
$\hint(\seq{x_h},\seq{x_\ell},z)$ to be  much simpler than $\formulaa(\seq{x_h},z)$, which reduces the number of quantifier alternations and significantly improves performance. 
In some cases we use hints of the form $ \hint(\seq{x_h},\seq{x_\ell},z) = ( \seq{x_\ell} 
 = \seq{t}(\seq{x_h},z) ) $, where $\seq{t}(\seq{x_h},z)$ are terms. 
 In such cases, totality is trivial and we have completely eliminated the existential quantification. 
In the TreeTermination example  we give a hint of this form for the transition \textsf{terminate}, see \Cref{fig:sq_for_tree}. 
The hint says that the transition $\mathsf{terminate}$ in high states with argument $\mathsf{n_h}$ 
is simulated in low states by %
the same argument, $\mathsf{n_\ell=n_h}$ (or by a stutter).

\subsection{Results}

We evaluate our approach on examples from various previous works, listed in \Cref{tab1}.
We use {\mypyvy} with Z3 version 4.12.2.0, run on a laptop running Windows, with a Core-i7 2.8 GHz CPU.
All run-times are under 10 seconds.

\Cref{tab1} summarizes our results. Each row describes one example. As described in \Cref{subsec:many_sorted}, for transition systems with multiple finite sorts we apply our approach iteratively, in which case each iteration is described using a separate row in the order of application. %
For each example, the columns list its sorts and their semantics, the sort for which we show a cutoff together with the %
cutoff bound, how many relation and function symbols are given update terms and formulas (instead of using default values), 
whether %
an invariant is given, and how many transitions are given hints.
For all examples the default formula for the deletion condition is used.
In the sorts column, ``$\sort s$ : fin'' indicates that $\sort s$ is an uninterpreted sort with the finite-domain semantics, ``$\sort s$ : $k$''  indicates that  $\sort s$ is an uninterpreted sort with the bounded semantics with bound $k$ (this is encoded by the axiom $\mathrm{size}_{\leq k}$), and
``\textsf{int}'' stands for the interpreted integer sort.

    \begin{table}[t]
\caption{Summary of Results.}\label{tab1}
\begin{tabular}{|l|l|l|c|c|c|}
\hline
Example &  Sorts \& Semantics  & Cutoff & Updates & Inv. & Hints \\
\hline
\hline
TreeTermination &  \textsf{node}: fin & \textsf{node}: 2 & 1/3 & $\checkmark$  & 1/1 \\
\hline

EchoMachine & \textsf{round}: fin, \textsf{value}: fin & \textsf{round}: 3 & 0/3 & $\checkmark$  & 0/2  \\
  & \textsf{round}: $3$, \textsf{value}: fin & \textsf{value}: 2 & 0/3 & $\times$  & 0/2 \\
\hline
ToyConsensus 
  & \textsf{value}: fin, \textsf{quorum}: fin, \textsf{node}: fin & \textsf{value}: 2 & 1/3 & $\checkmark$ & 0/2 \\
  & \textsf{value}: 2, \textsf{quorum}: fin, \textsf{node}: fin & \textsf{quorum}: 2 & 0/3 & $\times$ & 0/2 \\
  & \textsf{value}: 2, \textsf{quorum}: 2, \textsf{node}: fin & \textsf{node}: 4 & 0/3 & $\times$ & 1/2\\
\hline
LockServer &  \textsf{node}: fin & \textsf{node}: 2 & 1/5 & $\checkmark$  & 0/5 \\
\hline

ListToken & \textsf{node}: fin & \textsf{node}: 3 & 2/7 & $\checkmark$ & 1/3 \\
\hline

PlusMinus & \textsf{thread}: fin, \textsf{int} & \textsf{thread}: 1 & 1/3 & $\checkmark$ & 0/2 \\
\hline

EqualSum & \textsf{thread}: fin, \textsf{index}: fin, \textsf{int} & \textsf{thread}: 2 & 0/8 &  $\times$ & 0/1  \\
  & \textsf{thread}: $2$, \textsf{index}: fin, \textsf{int} & \textsf{index}: 2 & 3/8 & $\checkmark$ & 0/1 \\
\hline

EqualSumOrders & \textsf{thread}: fin, \textsf{index}: fin, \textsf{int} & \textsf{thread}: 2 & 0/8 &  $\times$ & 0/1  \\
  & \textsf{thread}: $2$, \textsf{index}: fin, \textsf{int} & \textsf{index}: 4 & 3/8 & $\checkmark$ & 0/1 \\
\hline

\end{tabular}
\end{table}

\medskip Next, we briefly describe each of the examples.

\textit{TreeTermination.} Our motivating example from \Cref{ex1}, originally from~\cite{distributed_network_protocols}. Crucially, this example cannot be proved by means of an inductive invariant as it is only safe under the finite-domain semantics.

\textit{EchoMachine.} Taken from \cite{infinite_models}. 
A simple distributed protocol where values are echoed in rounds which are linearly ordered. A value may be echoed in some round if there exists some previous round that echoed it. 
The safety property is that the value echoed in all rounds is unique. 
\cite{infinite_models} shows a simple inductive invariant that implies safety over finite structures but not in the first-order semantics; we use the same invariant in a high-low update to establish safety.

\textit{ToyConsensus.} Taken from~\cite{modularity_deductive}.
A simple consensus protocol with three finite sorts, including sort $\sort{quorum}$, which is axiomatized to ensure a nonempty intersection for any two quorums.
We highlight this example as an example where the order of application on sorts matters. Specifically, after we show a cutoff of $2$ for the sort \textsf{quorum}, we can  instantiate the quorum intersection axiom and  Skolemize the node in the intersection of any pair of quorums, giving $4$ Skolem constants. We can then freely delete any other node.

\textit{LockServer.} Taken from \cite{ivy}. A simple %
mutual exclusion protocol. In this example, our approach does not simplify the proof since to prove simulation the high-low update must include an invariant that suffices to prove  safety directly.  

\textit{ListToken.} A custom example where a token travels along a linked list, modeled with a $\mathsf{next}$ relation.
An invariant that proves safety would require talking about the transitive closure of $\mathsf{next}$, which cannot be done in FOL. Instead, our approach based on simulation avoids the need to reason about the transitive closure.
The high-low update changes the $\mathsf{next}$ relation to skip over the deleted node, and moves the token accordingly.

\textit{PlusMinus.} Taken from \cite{thread_modularity}.
A simple concurrent program where any number of threads each move between two program locations, incrementing and then decrementing a shared counter, which must always remain non-negative. 
In~\cite{thread_modularity}, it is shown that no Ashcroft invariant proves its safety. 
Intuitively, an inductive invariant would need to count the number of threads in some program location, which is problematic in FOL. Instead, the high-low update undoes increments to the counter by the deleted thread.

\textit{EqualSum.} Taken from \cite{commutativity_simplifies}.
A simple concurrent program where threads sum over elements of an array.
The safety property is that any two threads that terminate hold the same sum.
An inductive invariant would need to record sums of array segments of unbounded length, which is not clear how to express in FOL. 
Our approach easily establishes a cutoff of 2 for sort $\sort{thread}$. Then, to establish a cutoff for sort $\sort{index}$, the high-low update deletes an array element and subtracts it from the sum of any thread pointing to a higher index.

\textit{EqualSumOrders.} A variant of EqualSum where every thread holds a different order on the indices of the array. Unlike EqualSum, in this variant the order of application of our method on the sorts matters. To prove the cutoff for sort $\mathsf{index}$ we use a similar Skolemization to the one used in ToyConsensus.

\paragraph{Limitations.}
While our approach is applicable to some examples that cannot be verified by previous approaches, it has 
limitations.
The definition of a high-low update does not specify how the state should be updated 
to obtain a simulation relation. 
In our examples, we used the high-low update to undo all changes that involved the deleted element $z$ in all %
transitions so far. In this way, transitions that are completely undone are simulated by stutters.
For some examples, it is hard to keep track of the changes that need to be undone, and it requires some instrumentation.
For example, in the ToyConsensus example, when a value is chosen by some quorum we change the model to keep track of the quorum that chose it, and not only the fact that it was chosen. 
In some cases, undoing the changes that involved $z$ causes transitions of other elements to be impossible to simulate, for example, if the transition has a guard on some global variable, affected by $z$.

\section{Related Work}\label{sec:related}

Parameterized verification of safety properties is a widely-researched problem. Many works~\cite{invisible_invariants,smt_parameterized,swiss,thread_modularity,pdr_alterations,thread_modular_counter_abstraction,duoAI} have developed techniques to infer inductive invariants for parameterized systems in first-order logic. Such techniques are 
bound to fail if no satisfactory invariant exists in the language considered, or if the system is only safe in the finite instances. Our approach is deductive and based on the notion of cutoffs~\cite{view_abstraction,tight_cutoffs,many_to_few,reasoning_about_rings,quicksilver,better_cutoffs,dynamic_cutoff_detection,cutoff_concensus,namjoshi2007,cutoffs_point_to_point,agreement_based_systems}.  
Most cutoff results
guarantee that a cutoff always exists for a class of systems and properties. In turn, the system specification is much more limited. Our work proposes a framework for establishing cutoffs for a specific first-order transition system and property.
Next we focus on the works closest to ours.

\paragraph{Cutoffs by Simulation Relations.}  
In~\cite{bhat_nagar}, the cutoff method is utilized for safety verification of distributed protocols specified in FOL. They automatically generate and check a simulation relation between an arbitrary system instance and a cutoff instance.%
We, on the other hand, encode a simulation relation between system instances that differ in size by $1$, and not directly to the cutoff instance. In this way, our work captures an inductive argument over the state size, which lets it %
capture the finite-domain semantics. In contrast,~\cite{bhat_nagar} cannot verify examples such as \Cref{ex1} that are only correct in finite instances. Additionally, the simulation relation in~\cite{bhat_nagar}
is defined by mapping all elements in any instance to the elements of the cutoff instance such that 
transitions are preserved, while 
we use a partial mapping of elements
with updates that allow to undo transitions. %
Relying on a fixed total mapping of elements is bound to fail in examples 
such as PlusMinus, where threads move between program locations.
In~\cite{clarke_grumberg}, a cutoff result is given for a parameterized mutual-exclusion protocol by giving a simulation relation between system instances that differ in size by $1$. Their proof of simulation is manual, and uses an unbounded number of steps, which cannot be automated by an SMT solver.
\cite{squeezers} formulates the notion of a simulation-inducing squeezer for safety properties, which is essentially a \emph{deterministic} size-reducing simulation with a slightly weaker notion of stuttering. 
This work is evaluated only on deterministic sequential array programs and cannot capture transition systems defined in FOL.
Other works~\cite{difference_invariants,full_program_induction} automatically generate proofs for program assertions by induction on some size parameter.
These methods are currently limited to programs manipulating arrays.

\paragraph{Mitigating Limitations of FOL.}
\cite{VeriCon,heap_paths,padon_phd,paxos_made_epr,modularity_deductive,quantifier_instantitation} use the Effectively Propositional (EPR)~\cite{epr} decidable logic of FOL to model and verify transition systems. 
EPR has the finite-model property: any satisfiable formula has a finite model. In the case of safety verification, this has the effect of capturing the finite-domain semantics. However, it is not always clear how to ensure that the verification conditions are in EPR. 
When the verification conditions %
are outside of EPR, as in \Cref{ex1}, infinite models may arise. \cite{infinite_models} finds such models and shows that they can sometimes be eliminated by induction axioms. \cite{inductive_lemmas} develops an approach for automatically generating induction axioms for systems modeled using algebraic data-types. 
\cite{local_reasoning} takes an alternative approach to modeling systems in FOL, separating the specification of topology from the specification of the protocol.

\paragraph{Mitigating Limitations of Inductive Invariants.}
In~\cite{commutativity_simplifies}, 
programs whose inductive invariants are not expressible in FOL (or are too complex) are verified by finding 
an inductive invariant for a reduction of the program obtained by reordering execution steps, relying on commutativity.
Counting Proofs~\cite{proofs_that_count} create proofs that keep track of the number of threads that satisfy some properties, akin to examples like PlusMinus.

\section{Conclusion}\label{sec:conclusion}

We have presented a framework for specifying and validating size-reducing simulations in first-order logic as a way to establish cutoffs for the finite-domain semantics of first-order transition systems.
Size-reducing simulations mimic induction over the domain size, which is sometimes simpler than induction over time, and, notably, leverages the finite-domain semantics.

As our evaluation shows, our framework manages to verify some examples that cannot be verified by previous approaches. 
We leave automating the search for a high-low update to future work. 
In this context we note that without hints our verification conditions are not in EPR, and therefore SMT solvers sometimes struggle to find models when a query is not valid, which might be an obstacle for automation.

It is possible to extend our approach to a larger class of size-reducing simulations,
e.g., by allowing to delete more than one domain element, or relaxing the deterministic updates to other syntactic restrictions that still guarantee \Cref{lem:totality}. We did not find such extensions necessary in the examples we considered.
Another interesting direction for future research is to consider a wider range of properties, including liveness properties. While the definition of a size-reducing simulation can be extended naturally, the more intricate structure of traces that satisfy liveness properties seems to be a challenge.

\subsubsection{Acknowledgement}
We thank Neta Elad, Shachar Itzhaky and Orna Grumberg for helpful advice.
The research leading to these results has received funding from the
European Research Council under the European Union's Horizon 2020 research and innovation programme (grant agreement No [759102-SVIS]).
This research was partially supported by the Israeli Science Foundation (ISF) grant No.\ 2117/23.

\bibliography{references}

\newpage

\appendix

\section{Proofs}\label{appendix:proofs}

    \setcounter{lemma}{0}
    \setcounter{theorem}{0}

\begin{theorem}\label{theorem:soundness_app}
    If $\simul$ is a size-reducing simulation for $\Tpar$, $P$  and $k$, then $k$ is a cutoff for $\Tpar$ and $P$.
\end{theorem}
\begin{proof}
    We show that $\Tpar \models P$ if and only if $T_\domain \models P$ for every $\domain$ with $|\domain|\leq k$. 
    The direction from left to right follows directly from the definition of $\Tpar$. 
    For the opposite direction, suppose $T_\domain \models P$ for every $\domain$ with $|\domain| \leq k$ and
    assume towards a contradiction that $\Tpar \not \models P$. Then there is a reachable state $s$ of $\Tpar$ with $s\notin P$. 
    Take such a reachable state $s^h$ such that  $s^h\notin P$ with minimal $n=|\struct^h|$ (since all the states in $\Tpar$ are finite, minimality is well-defined). Denote $\domain = \dom(s^h)$.
    The trace showing $s^h$ is reachable is a trace of $T_\domain$. If $n\leq k$ we get a contradiction to $T_\domain$ being safe.
   
    Otherwise, $|\struct^h| = n > k$, $s^h$ is reachable so we have a trace $(s^h_i)_{i=0}^t$ with $s^h_t = s^h$. 
    We will construct, by induction on $i$, a reachable state $s^\ell_i$ such that $(s^h_i,s^\ell_i)\in \simul$. For $i=0$ we have $s_0^h\in \init_\para$, by \cref{item:initation} we have a state $s^\ell_0 \in \init_\para$ such that $(s^h_0,s^\ell_0)\in \simul$ and initial states are reachable. Assume we've constructed a reachable state $s^\ell_i$ such that $(s^h_i,s^\ell_i)\in \simul$, for $i+1$ we have $(s^h_i,s^h_{i+1})\in R_\para$. By \cref{item:simulation}  we have a state $s^\ell_{i+1}$ such that $(s^h_{i+1},s^\ell_{i+1})\in \simul$ and $(s^\ell_i,s^\ell_{i+1})\in R_\para^\star$, it follows that $s^\ell_{i+1}$ is reachable. Denote $s^\ell=s^\ell_t$,  $s^\ell$ is reachable and $(s^h,s^\ell)\in \simul$. We have $s^h\notin P$ so by \cref{item:bad}, we have $s^\ell\notin P$. By \cref{item:size}, we have $|\struct^\ell| < |\struct^h|$, in contradiction to minimality of $|\struct^h|$. \qed

\end{proof}

\begin{lemma}\label{lem:strong_is_weak_app}
    A strong size-reducing simulation is a size-reducing simulation. 
\end{lemma}
\begin{proof}
        The properties of size-reduction and fault-preservation are the same in both definitions. We will show that inductive totality along with strong initiation and strong simulation imply initiation and simulation of \Cref{def:size_reducing_sim}.
        \paragraph{Initiation.} Let $s^h\in \init_\para$ with $|s^h|>k$. By inductive totality we have a state $s^\ell\in S_\para$ such that $(s^h,s^\ell)\in H$. By strong initiation we get that $s^\ell\in \init_\para$.
        \paragraph{Simulation.} Let 
        $(s^h,s^\ell)\in H$ and $s^\hpr$ such that 
        $(s^h,s^{h\prime})\in R_\para$. By inductive totality we have a state $s^\ellpr \in S_\para$ such that $(s^\hpr,s^\ellpr)\in H$ and $\dom(s^\ell)=\dom(s^\ellpr)$. By strong simulation it follows that $ (s^\ell,s^\ellpr)\in R_\para$ or $s^\ell = s^\ellpr$, and in either case $(s^\ell,s^\ellpr)\in R_\para^\star$.
    \qed 
\end{proof}

\begin{lemma}\label{lem:size_reduction_app}
    For a formula $\sq(z)$ over $\Sigma^h\uplus\Sigma^\ell$, and a pair $(s^h,s^\ell)\in \simul_\sq$ such that $s^h$ is finite, we have $|s^h| > |s^\ell|$.
\end{lemma}
\begin{proof}
    $s^\ell$ is an $\sq$-substructure of $s^h$, therefore there exists $d_0\in \dom(s^h)$ such that $\dom(s^\ell)=\dom(s^h)\setminus \{d_0\}$. Thus, $|s^\ell | = |s^h| - 1 < |s^h|$. \qed
\end{proof}

\newpage
\noindent For the next proofs we need some notation: for a structure $s=(\domain,\interp)$, a term $t$ with free variables $\Vars(t)$ and an assignment $\assign : \Vars(t)\to \domain$ we denote by $\terminterp{t}{s}{\assign}$ the interpretation of $t$ in $s$ under $\assign$. 
For variables $x_1,\ldots,x_m$ and domain elements $d_1,\ldots,d_m$ we denote by $[\seq x \mapsto \seq d]$ the assignment $\assign$ with $\assign(x_i)=d_i$ for $1\leq i \leq m$. For an assignment $\assign$, a variable $x$, and a domain element $d$, we denote by $\assign[x\mapsto d]$ the assignment that agrees with $\assign$ on all variables except for $x$, which is mapped to $d$.

\begin{lemma}\label{lem:domain_extension_app}
Let $\formula$ be a closed formula over $\Sigma$ and $z\notin \Vars(\formula)$ a variable.
Let $(\domain,\interp)$ be a structure for $\Sigma$, $d_0\in \domain$ and $(\exc\domain,\exc\interp)$  a substructure of $(\domain,\interp)$ such that $\exc \domain = \domain \setminus \{d_0\} $. Then %
    $
    (\exc \domain,\exc \interp) \models \formula \iff (\domain, \interp),[z\mapsto d_0] \models \ADrep{z}\formula.
    $
\end{lemma}

\begin{proof}  
    We first strengthen the lemma and allow $\formula$ to be any formula (not necessarily closed). In the same setting, for any assignment $\assign$ of $\FV(\formula)$ to $\exc \domain$ we claim that
    $
    (\exc \domain,\exc \interp),\assign \models \formula \iff (\domain, \interp),\assign[z\mapsto d_0] \models \ADrep{z}\formula.
    $
    
    The proof goes by induction on the structure of $\formula$. First, we notice that if $t$ is a term with $z\notin \Vars(t)$ and $\assign$ is an of assignment of $\Vars(t)$ to $\exc\domain$ then $\terminterp{t}{(\domain,\interp)}{\assign[z\mapsto d_0]} = \terminterp{t}{(\exc \domain, \exc \interp)}{\assign} \in \exc \domain  
$
, which can be shown by induction on the structure of the term, relying on the property that $(\exc \domain,\exc \interp)$ is a substructure of $(\domain, \interp)$,
and hence interpretations of constant and function symbols coincide.
All terms that appear in $\formula$ never have $z$ as a free variable, so it follows that they are interpreted in $(\domain,\interp),\assign[z\mapsto d_0]$ to elements of $\exc\domain$.
    
The case where $\formula = (t_1=t_2)$ follows easily. We use this as well for the case where $\formula = r(t_1,\ldots,t_m)$:
    \begin{align*}
        (\exc \domain,\exc \interp),\assign \models r(t_1,\ldots,t_m) & \iff\\ ( \terminterp{t_1}{(\exc\domain, \exc\interp)}{\assign} ,\ldots, \terminterp{t_m}{(\exc\domain, \exc\interp)}{\assign}) \in \exc\interp(r)& \iff \\
        (\terminterp{t_1}{(\domain, \interp)}{\assign[z\mapsto d_0]} ,\ldots, \terminterp{t_m}{(\domain, \interp)}{\assign[z\mapsto d_0]}) \in  \interp(r) \cap \exc \domain^m& \iff\\
        (\terminterp{t_1}{(\domain, \interp)}{\assign[z\mapsto d_0]} ,\ldots, \terminterp{t_m}{(\domain, \interp)}{\assign[z\mapsto d_0]}) \in  \interp(r) &\iff\\
        (\domain, \interp),[z\mapsto d_0] \models r(t_1,\ldots,t_m)
    \end{align*}  
    First equivalence is the semantics of relations, the second is definition of a substructure and the observation above, the third is the observation above, and the last equivalence is again the semantics of relations.

    For the case where $\formula = \forall x. \formulaa$, then $ \ADrep{z}{\formula} = \forall x.\ (x\neq z) \to \ADrep{z}{\formulaa}$.
    \begin{align*}
        (\exc \domain,\exc \interp),\assign \models \forall x. \formulaa &\iff\\
        \forall d\in \exc\domain. (\exc\domain,\exc\interp),\assign[x\mapsto d]\models \formulaa &\iff\\ 
        \forall d\in \exc\domain. (\domain,\interp),\assign[x\mapsto d,z\mapsto d_0]\models \ADrep{z}{\formulaa} &\iff\\
         \forall d\in \domain. ( \domain, \interp), \assign[z\mapsto d_0,x\mapsto d]\models (x\neq z) \to \ADrep{z}{\formulaa} &\iff\\
          ( \domain, \interp),\assign[z\mapsto d_0] \models \forall x. (x\neq z)\to \ADrep{z}{\formulaa}
    \end{align*}  
    First equivalence is the semantics of universal quantification, the second is the induction hypothesis, the third is seen to be true by dividing into the cases $d = d_0$ and $d\neq d_0$, and noticing $\assign[x\mapsto d,z\mapsto d_0]=\assign[z\mapsto d_0,x\mapsto d]$, as $x$ and $z$ are different variables by the assumption $z\notin \Vars(\formula)$. The last equivalence is again the semantics of universal quantification.

    The cases for negation and conjunction are trivial. Disjunction and Existential quantification follow from the previous cases.

\qed
\end{proof}

    Throughout the following proofs we use abuse of notation thinking of interpretations of states such as $s^h,s^\ell$, which are  interpretations of $\Sigma$ as interpretations of $\Sigma^h,\Sigma^\ell$ and so on, when necessary, as discussed in \Cref{sec:prelims}. 
    We also freely use the fact that if $\alpha$ is a formula over vocabulary $\Sigma^h$ and and $\interp^h,\interp^\ell$ are interpretations for domain $\domain$ and vocabularies $\Sigma^h,\Sigma^\ell$ respectively, such that $(\domain,\interp^h\uplus\interp^\ell)\models \alpha$ then $(\domain,\interp^h)\models \alpha$.

\begin{lemma} \label{lem:vc_simulation_app}
Let $\sq(z)$ be a formula over $\Sigma^h\uplus \Sigma^\ell$. If verification conditions (\ref{item:vc-iota})--(\ref{item:vc-safety}) 
of \Cref{fig:verificationsconditions} are valid then $\simul_{\sq}\cap (S_\para \times S_\para)$ satisfies \cref{item:strong_initation,item:strong_simulation,item:strong_bad} of \Cref{def:strong_size_reducing_sim}.
\end{lemma}

\begin{proof}
    We prove strong initiation and strong simulation, fault preservation is similar to strong initiation. 
    
    \paragraph{Strong Initiation:} Let $(s^h,s^\ell)\in \simul_\sq \cap (S_\para \times S_\para)$ such that $s^h \in \init_\para$, we show that $s^\ell \in \init_\para$.
    From \Cref{def: relation_from_sq} we have $s^h = (\domain,\interp^h)$, $s^\ell = (\exc \domain, \exc \interp^\ell)$ and $d_0\in \domain$ such that $\exc \domain = \domain \setminus \{d_0\}$ and an interpretation $\interp^\ell$ such that $(\exc\domain,\exc\interp^\ell)$ is a substructure of $(\domain,\interp^\ell)$ and $(\domain,\interp^h\uplus \interp^\ell),[z\mapsto d_0] \models \sq(z)$. $s^h$ is an initial state so we have $(\domain,\interp^h)\models \Gamma^h \wedge \iota^h$. By validity of verification condition (\ref{item:vc-iota}), we have $(\domain,\interp^\ell),[z\mapsto d_0]\models \ADrep z \iota^\ell$. From \Cref{lem:domain_extension_app} we get $(\exc \domain, \exc\interp^\ell) \models \iota^\ell$ and so $s^\ell \in \init_\para$ as required.
        
    \paragraph{Strong Simulation:}
    Let $(s^h,s^\ell), (s^\hpr,s^\ellpr)\in \simul_\sq \cap (S_\para\times S_\para)$
    such that $(s^h,s^\hpr)\in R_\para$ and $\dom(s^\ell)=\dom(s^\ellpr)$. We show that $(s^\ell,s^\ellpr)\in R_\para$ or $s^\ell = s^\ellpr $.
    From $(s^h,s^\hpr)\in R_\para$ it follows that $\dom(s^h)=\dom(s^\hpr)$. Along with $\dom(s^\ell)=\dom(s^\ellpr)$, unfolding  \Cref{def: relation_from_sq} gives:
    a high domain $\domain$ and low domain $\exc \domain$ with $d_0\in\domain$ such that $\exc \domain = \domain \setminus \{d_0\}$, and interpretations:  $\interp^h,\interp^\ell,\exc \interp^\ell, \interp^{h\prime},\interp^{\ell\prime}, \exc \interp^{\ell\prime}$ such that:
    $(\exc \domain, \exc \interp^{\ell})$ is a substructure of $(\domain,\interp^{\ell})$ and $(\exc \domain, \exc \interp^{\ell\prime})$ is a substructure of $(\domain,\interp^{\ell\prime})$ and 
    $(\domain,\interp^h\uplus\interp^\ell),[z\mapsto d_0]\models \sq(z)$, 
    $(\domain,\interp^h\uplus\interp^{\ell\prime}),[z\mapsto d_0]\models \sq'(z)$.
    By definition of states and transitions we have:
    $(\domain,\interp^h) \models \Gamma^h,(\domain,\interp^{h\prime}) \models \Gamma^{h\prime}$ and $(\domain, \interp^h \uplus \interp^{h\prime} )\models \tau^h $. By verification condition (\ref{item:vc-tau}) we have:
    $(\domain, \interp^\ell \uplus \interp^{\ell\prime}),[z\mapsto d_0] \models \ADrep z {(\tau\vee \idle)^\ell} $. 
    By \Cref{lem:domain_extension_app} we get: $(\exc \domain, \exc \interp^\ell \uplus \exc \interp^{\ell\prime}) \models (\tau\vee \idle)^\ell $. It follows that $(s^\ell,s^{\ell\prime})\in R_\para$ or $s^\ell = s^{\ell\prime}$. \qed
\end{proof}

\begin{lemma}\label{lem:totality_app}
    Let $\sq(z)$ be a high-low update with precondition $\theta(z)$ and $(\domain,\interp^h)$  a structure with $d_0\in \domain$ such that $(\domain,\interp^h),[z\mapsto d_0]\models \cond(z)$. Then there exists an interpretation $\interp^\ell$ such that 
    $(\domain,\interp^h\uplus\interp^\ell),[z\mapsto d_0] \models \sq(z)$.\end{lemma}

\begin{proof}
    We define the interpretation 
$\interp^\ell$ by: for every relation $r\in \Sigma$  define $\interp^\ell(r^\ell) = \{ \seq d \mid (\domain,\interp^h),[z\mapsto d_0, \seq x\mapsto \seq d] \models \formula_r(\seq x,z)\} $, and for every function $f\in \Sigma$ define:
    $ \interp^\ell(f^\ell)(\seq d) = 
    {(t_f)}^{(\domain,\interp^h), [z\mapsto d_0, \seq x\mapsto \seq d]} $.   
    That we have $(\domain,\interp^h\uplus \interp^\ell),[z\mapsto d_0]\models \sq(z)$ is now trivial.\qed
\end{proof}

\begin{lemma}\label{lem:vc_inductive_totality_app}
    Let $\sq(z)$ be a high-low update with precondition $\theta(z)$. If verification conditions (\ref{item:vc-project})--(\ref{item:vc-muconsec}) of 
    \Cref{fig:verificationsconditions} are valid,  
    then $\simul_\sq \cap (S_\para \times S_\para)$ satisfies \cref{item:strong_inductive_totality} of \Cref{def:strong_size_reducing_sim}.
\end{lemma}

\begin{proof}
We show inductive totality.
\paragraph{Initial States.} For an initial state $s^h \in \init_\para$ with $|s^h|>k$, we show the existence of $s^\ell$ such that $(s^h,s^\ell)\in \simul_\sq \cap (S_\para\times S_\para)$. We have $s^h \models \Gamma^h \wedge \iota^h \wedge \size$.
From verification condition (\ref{item:vc-muinit}) we get $s^h\models \exists z. \cond(z)$.
Denote $s^h = (\domain, \interp^h)$, then there is $d_0\in \domain$ such that $s^h,[z\mapsto d_0]\models \cond(z)$. From \Cref{lem:totality_app} there is an interpretation $\interp^\ell$ such that $(\domain,\interp^h\uplus \interp^\ell),[z\mapsto d_0] \models \sq(z)$.
Define the domain $\exc \domain = \domain \setminus \{d_0\}$ and the interpretation $\exc \interp^\ell$ by: for every relation symbol $r\in \Sigma$ with arity $m$: $\exc \interp^\ell(r^\ell) = \interp^\ell(r^\ell)\cap \exc \domain^m $ and for every function symbol $f\in \Sigma$ with arity $m$: $\exc \interp^\ell (f^\ell) = \interp^\ell(f^\ell) |_{\exc \domain^m}$. By verification condition (\ref{item:vc-project}), we have $(\domain,\interp^\ell) \models \closed z f$ and so $\exc \interp^\ell(f)$ is a function to $\exc \domain^m$ for all $f\in \Sigma$ so $s^\ell=(\exc \domain,\exc\interp^\ell)$ is well-defined as a structure. From the above we get $(s^h,s^\ell)\in \simul_\sq$. By verification condition (\ref{item:vc-gamma}) we have $(\domain,\interp^\ell)\models \ADrep z {\Gamma}^\ell$.
From \Cref{lem:domain_extension_app} we get $(\exc \domain, \exc \interp^\ell) \models \Gamma^\ell$, and so $s^\ell \in S_\para$, and we have $(s^h,s^\ell)\in \simul_\sq \cap (S_\para\times S_\para)$.
    
\paragraph{Transitions.} Assume $(s^h,s^\hpr)\in R_\para$ and $s^\ell$ such that $(s^h,s^\ell)\in \simul_\sq \cap (S_\para \times S_\para)$. We show the existence of $s^\ellpr$ such that $(s^\hpr,s^\ellpr)\in \simul_\sq \cap (S_\para \times S_\para)$. Denote $s^h=(\domain,\interp^h), s^\ell =( \exc \domain,\exc \interp^\ell)$. From \Cref{def: relation_from_sq}, we get $d_0\in\domain$ such that $\exc \domain = \domain \setminus \{d_0\}$ and an interpretation $\interp^\ell$ such that $(\domain,\interp^h \uplus \interp^\ell),[z\mapsto d_0],\models \sq(z)$. It follows that $(\domain,\interp^h),[z\mapsto d_0] \models \cond(z)$. Denote $s^\hpr = (\domain, \interp^\hpr)$, we then have $(\domain,\interp^h\uplus\interp^\hpr), [z\mapsto d_0] \models \Gamma^h \wedge \Gamma^\hpr \wedge \tau^h \wedge \cond(z)$. From verification condition (\ref{item:vc-muconsec}) we have $(\domain,\interp^\hpr),[z\mapsto d_0]\models \cond'(z)$. From \Cref{lem:totality_app} we get an interpretation $\interp^\ellpr$ such that $(\domain,\interp^\hpr\uplus\interp^\ellpr),[z\mapsto d_0] \models \sq(z)$, from which we construct an interpretation $\exc\interp^\ellpr$ similarly to the previous item. We then have $s^\ellpr=(\exc\domain,\exc\interp^\ellpr) \in S_\para$, $\dom(s^\ell)=\dom(s^\ellpr)=\exc \domain$ and finally $(s^\hpr,s^\ellpr)\in \simul_\sq\cap (S_\para\times S_\para)$. \qed
\end{proof}

\begin{theorem}\label{theo:summary_app}
If all the verification conditions for a high-low update $\sq(z)$, $\Tspec$, $\safetyspec$ and $k$ 
(\Cref{fig:verificationsconditions}) are valid, then 
$\simul_\sq\cap (S_\para\times S_\para)$ is a strong size-reducing simulation for $\Tpar$, $P$ and $k$. Consequently, $k$ is a cutoff for $\Tpar$ and $P$.
\end{theorem}\begin{proof}   
    Corollary of \Cref{theorem:soundness_app,lem:size_reduction_app,lem:vc_simulation_app,lem:vc_inductive_totality_app}. \qed

\end{proof}

\end{document}